\pdfoutput=1
\documentclass[acmsmall, table, rgb, nonacm, natbib=false, screen]{acmart}

\usepackage{subcaption}
\usepackage{tikz}
\usetikzlibrary{backgrounds,fit,decorations.pathreplacing,calc,matrix,arrows.meta,positioning,decorations.pathreplacing,decorations.pathreplacing,calligraphy}
\usepackage{newtxmath}
 \usepackage[textsize=tiny]{todonotes}
 \usepackage{nicefrac}
 \usepackage{csvsimple}
 \usepackage{svg}
\usepackage{csquotes}
\usepackage{mathtools}
\usepackage{enumitem}
\usepackage[binary-units=true, detect-all=true]{siunitx}
\usepackage{url}
\usepackage{graphicx}
\usepackage{setspace}
\usepackage{multirow, makecell}
\usepackage{array}
\usepackage[braket, qm]{qcircuit}
\usepackage{marvosym}
\usepackage{xcolor}
\usepackage{verbatim}
\usepackage{nicematrix}
\usepackage{algorithm}
\usepackage{algpseudocode}
\usepackage{flushend}
\usepackage{amsfonts}
\usepackage{soul}
\usepackage{etoolbox}
\makeatletter
\patchcmd{\SOUL@ulunderline}{\dimen@}{\SOUL@dimen}{}{}
\patchcmd{\SOUL@ulunderline}{\dimen@}{\SOUL@dimen}{}{}
\patchcmd{\SOUL@ulunderline}{\dimen@}{\SOUL@dimen}{}{}
\newdimen\SOUL@dimen
\makeatother

\algblock{Input}{EndInput}
\algnotext{EndInput}
\algblock{Output}{EndOutput}
\algnotext{EndOutput}
\newcommand{\Desc}[2]{\State \makebox[5em][l]{#1}#2}

\usepackage[style=ieee, 
	maxnames=3, 
	minnames=1, 
	sortlocale=en_US,
	sortcites=true,
	url=true,
	eprint=true,
	giveninits=false,
	backend=biber,
	dashed=false
	]{biblatex}
\AtEveryBibitem{%
  \clearname{editor}%
  \clearfield{series}%
  \clearfield{isbn}%
  \clearfield{issn}%
  \clearfield{day}%
  \clearfield{month}%
  \clearfield{place}%
  \clearlist{location}%
  \ifentrytype{inproceedings}{%
  \clearlist{publisher}
  }{}%
  \ifentrytype{software}{%
  }{%
  \clearfield{url}%
  \clearfield{urlyear}%
  }%
}

\makeatletter
\let\oldtheequation\theequation
\renewcommand\tagform@[1]{\maketag@@@{\ignorespaces#1\unskip\@@italiccorr}}
\renewcommand\theequation{(\oldtheequation)}
\makeatother

\newtheorem{example}{Example}

\newtheorem{definition}{Definition}
\newtheorem*{problem}{Central Problem}

\renewcommand*{\proofname}{Proof Sketch}

\authorsaddresses{}

\newcommand{\swap}{\operatorname{SWAP}}
\newcommand{\cnot}{\operatorname{CNOT}}

\newcommand{\cand}{\operatorname{Cand}}

\newcommand{\sopt}{s_{\text{opt}}}

\newcommand{\cov}{\operatorname{Cov}}

\definecolor{myyellow}{RGB}{255, 253, 0}
\definecolor{myorange}{RGB}{238, 135, 51}
\definecolor{myblue}{RGB}{47, 112, 137}

\newcounter{relctr} %
\everydisplay\expandafter{\the\everydisplay\setcounter{relctr}{0}} %

      \AtBeginDocument{} %

\emergencystretch=\maxdimen
\begin{document}

\title{On Optimal Subarchitectures for Quantum Circuit Mapping}
\author{Tom Peham}
\orcid{0000-0003-3434-7881}
\affiliation{
  \department{Chair for Design Automation}
  \institution{Technical University of Munich}
  \city{Munich}
  \country{Germany}
}
\email{tom.peham@tum.de}
\author{Lukas Burgholzer}
\orcid{0000-0003-4699-1316}
\affiliation{%
    \department{Institute for Integrated Circuits}
  \institution{Johannes Kepler University}
  \city{Linz}
  \country{Austria}
}
\email{lukas.burgholzer@jku.at}
\author{Robert Wille}
\orcid{0000-0002-4993-7860}
\affiliation{%
  \institution{Technical University of Munich}
  \department{Chair for Design Automation}
  \city{Munich}
  \country{Germany}
}
\affiliation{%
  \institution{Software Competence Center Hagenberg GmbH}
  \city{Hagenberg}
  \country{Austria}
}
\email{robert.wille@tum.de}

	\begin{abstract}
          Compiling a high-level quantum circuit down to a low-level description that can be executed on
          state-of-the-art quantum computers is a crucial part of the software stack for quantum computing. One step in
          compiling a quantum circuit to some device is quantum circuit mapping, where the circuit is transformed such
          that it complies with the architecture’s limited qubit connectivity. Because the search space in quantum
          circuit mapping grows exponentially in the number of qubits, it is desirable to consider as few of the
          device’s physical qubits as possible in the process. Previous work conjectured that it suffices to consider
          only subarchitectures of a quantum  computer composed of as many qubits as used in the circuit. In this work,
          we refute this conjecture and establish criteria for judging whether considering  larger parts of the
          architecture might yield better solutions to the mapping problem. We show
          that determining subarchitectures that are of minimal size, i.e., from which no physical qubit can be removed
          without losing the optimal mapping solution for some quantum circuit, is a very hard problem. Based on a
          relaxation of the criteria for optimality, we introduce a relaxed consideration that still maintains
          optimality  for practically relevant quantum circuits. Eventually, this results in two methods for computing
          near-optimal sets of subarchitectures---providing the basis for \emph{efficient} quantum circuit mapping
          solutions. 
          We demonstrate the benefits of this novel method for state-of-the-art quantum computers by IBM, Google and Rigetti.
	\end{abstract}

        \maketitle
                \vspace*{-1mm}
\section{Introduction}
                \vspace*{-1mm}

        Quantum computing~\cite{nielsenQuantumComputationQuantum2010} is an emerging technology where
        computations are governed by quantum-mechanical principles. 
        Despite its relative infancy, even currently available quantum computers impose many challenges on
        quantum algorithm designers---increasing the need for design automation methods in the realm of quantum
        computing. 
        One critical challenge (for devices based on superconducting qubits~\cite{devoretSuperconductingCircuitsQuantum2013}) arises from the limited connectivity of the physical qubits on these quantum computers.
        Instead of allowing multi-qubit gates between arbitrary qubits, a \emph{coupling map} dictates which pairs of qubits may interact with each other over the course of a computation.
        Any quantum circuit executed on an actual device needs to conform to those restrictions.
		Thus, similar to classical computing, high-level quantum algorithms need to be \emph{compiled} to the
                target architecture---encompassing  many individual steps, such as synthesis, qubit
                allocation/placement, qubit routing, optimization, and
                scheduling~\cite{saeediSynthesisQuantumCircuits2011,siraichiQubitAllocation2018,zhuExactQubitAllocation2020,cowtanQubitRoutingProblem2019,namAutomatedOptimizationLarge2018,kissingerReducingTcountZXcalculus2020,muraliNoiseadaptiveCompilerMappings2019,burgholzerLimitingSearchSpace2022,willeMappingQuantumCircuits2019,zulehnerEfficientMethodologyMapping2019,zulehnerCompilingSUQuantum2019,hillmichExploitingQuantumTeleportation2021}.
		
		In the following, we focus on the steps dictated by the limited connectivity of the devices, i.e., qubit allocation/placement and routing, which is commonly referred to as \emph{quantum circuit mapping}.
		Mapping a quantum circuit first requires \emph{allocating} the circuit's logical qubits on the device's physical qubits, i.e., identifying a subarchitecture of the whole device where the circuit shall be executed.
		Then, the logical qubits need to be \emph{routed} on the physical qubits so that any operation in the circuit is executed on qubits that are connected on the device's coupling map.	
		This is commonly conducted by inserting $\swap$ gates into the circuit which allow dynamically changing the logical-to-physical qubit assignment~\cite{saeediSynthesisQuantumCircuits2011,siraichiQubitAllocation2018,zhuExactQubitAllocation2020,cowtanQubitRoutingProblem2019,muraliNoiseadaptiveCompilerMappings2019,burgholzerLimitingSearchSpace2022,willeMappingQuantumCircuits2019,zulehnerEfficientMethodologyMapping2019,zulehnerCompilingSUQuantum2019,hillmichExploitingQuantumTeleportation2021}.

        In order to ensure reliable execution of the resulting circuit, it is crucial to keep the overhead introduced through routing as small as possible.
		Unfortunately, determining optimal mappings of quantum circuits is an NP-hard
                problem~\cite{boteaComplexityQuantumCircuit2018, siraichiQubitAllocation2018}---mainly due to the
                involved search space growing exponentially in the number of considered qubits.
		This is exacerbated by the fact that quantum computers are not available in arbitrary sizes. 
		At the time of writing, e.g., IBM hosts systems using $1,5,7,27,65,127,$ and $433$ qubits.
		This means that, in many cases, significantly more qubits are available on an architecture than are
                strictly needed for mapping a quantum circuit. This is going to become even more problematic as the
                size-gap between newer generations of quantum computers increases. 
		For example, IBM plans to provide a quantum computer with $1121$ qubits in
                2023~\cite{gambettaIBMRoadmapScaling2020}. 
		A quick solution for this problem is to just pick a subarchitecture of the device that contains as many qubits as needed.
		While this keeps the search space for qubit routing as small as possible, it could be possible that incorporating more qubits than needed---effectively increasing the search space---allows for solutions requiring less overhead.
		Experimental results on optimal quantum circuit mapping conjectured that this is not the case~\cite{burgholzerLimitingSearchSpace2022,willeMappingQuantumCircuits2019}---suggesting that optimality is preserved by restricting the search space to the bare minimum.

		In this work, we answer this question by showing that considering larger
                subarchitectures can indeed yield better mappings. 
		Naturally, the next step is to ask which subarchitectures \emph{should} be considered during mapping so
                that no essential parts of the search space are cut off. Unsurprisingly, determining subarchitectures
                that are as small and that contain \emph{all} optimal mapping solutions for any quantum circuit of a
                specified size turns out to be a very hard problem. Luckily, the general notion of such \emph{optimal
                  subarchitectures} is hardly needed for mapping practical quantum circuits.
                This work defines and analyzes minimal sets of subarchitectures that allow for
                optimal mapping solutions for all but the most esoteric quantum circuits with a certain number of qubits,
                while drastically reducing the search space in many cases. To this end, the subarchitectures are ranked
                according to their coverage of optimality, i.e., the number of quantum circuits that can be optimally
                mapped to a given subarchitecture.  
		
		The resulting software is integrated into the open-source tool QMAP (available at \url{https://github.com/cda-tum/qmap}), which is part of the \emph{Munich Quantum Toolkit} (MQT), and can be used to compute various sets of
                subarchitectures with a trade-off between complexity and coverage of optimality.
        For convenience, it includes a pre-computed library of subarchitectures for common quantum computing devices. 
        Our results considering state-of-the-art quantum computers demonstrate the benefits of using the generated sets as the basis for quantum circuit mappers---allowing to significantly reduce the overall complexity of the search space%
		---while still ensuring that essential parts of the
                search space are covered for most quantum circuits.

        The rest of this paper is structured as follows. \autoref{sec:background} provides the necessary background
        in graph theory and quantum circuit mapping required to present the ideas in this work. After
        \autoref{sec:considered-problem} introduces the
        problem this work aims to address, \autoref{sec:better-mappings} defines two
        criteria that allow for the existence of better mapping solutions on larger subarchitectures and gives a
        complete characterization of optimal subarchitectures. Two algorithms for
        constructing \mbox{near-optimal} subarchitectures are presented in
        \autoref{sec:methods}. \autoref{sec:experiments} demonstrates the benefits of this novel method for existing
        quantum computing architectures. Finally, \autoref{sec:conclusion} concludes this paper.

        \vspace*{-2mm}
	\section{Background}\label{sec:background}
        The connectivity of qubits in modern NISQ~\cite{preskillQuantumComputingNISQ2018} devices is described via the
        architecture's coupling graph. Therefore, a lot of jargon related to graph theory will be used in the
        discussions on optimal \mbox{subarchitectures}. This section briefly introduces the relevant graph-theoretic
        concepts as well as the basics of quantum circuit mapping. For a comprehensive introduction to graph theory
        the reader is referred to any standard textbook such as~\cite{bondyGraphTheory}.

        \vspace*{-2mm}
              \subsection{Graph Theory}\label{sec:graph-theory}

        A \emph{graph} $A = (V, E)$ is comprised of a set of \emph{vertices} $V$ and a set of \emph{edges} $E \subseteq V
        \times V$ connecting vertices. The size of a graph is the number of vertices in the graph denoted $|A| =
        |V|$. $A$ is called \emph{undirected} if the pair $(v, w) \in E$ is unordered, otherwise it is called directed.
        
        Two graphs $A = (V_A, E_A)$, $A' = (V_{A'}, E_{A'})$ are \emph{isomorphic} (denoted $A \cong A'$) if there is a bijective function $h: V_A
        \rightarrow V_{A'}$ such that $(h(v), h(w)) \in E_{A'}$ iff $(v, w) \in E_A$.

        A graph $A' = (V_{A'}, E_{A'})$ is a \emph{subgraph} of $A = (V_A, E_A)$ (denoted $A' \sqsubseteq A$) if $V_{A'} \subseteq V_A$
        and $E_{A'} \subseteq E_A$. If $A'$ is a proper subgraph, i.e., $A' \neq A$ we write $A' \sqsubset A$. $A'$ is \emph{subgraph isomorphic}
        to $A$ if there is a subgraph $A'' \sqsubseteq A$ such that $A' \cong A''$.
        An \emph{induced subgraph} of a vertex set $V' \subseteq V$ in $A = (V, E)$ is a subgraph $A' \sqsubseteq A$ with $A' = (V', E')$ and the
        property that, if any vertices $v, w \in V'$ are connected in $A$, then $(v, w) \in E'$. In an overloading of notation we
        also call a subgraph $A' = (V', E')$ of $A$ an induced subgraph, if $A'$ is the induced subgraph of $V'$ in $A$.

        A \emph{path} in a graph $A = (V, E)$ is a (non-empty) sequence of distinct vertices \mbox{$p = (v_0, v_1, \cdots, v_{n-1})$} with $v_i \in V, 0
        \leq i < n$ and \mbox{$(v_i, v_{i+1}) \in E$}, $0 \leq i < n$. The length of this path is denoted $|p|$. For two vertices
        $v, w \in V$, $p_A(v, w)$ gives a shortest path between $v$ and $w$ in $A$. 
        The longest shortest path in a graph is called the graph's \emph{diameter}.

                \vspace*{-2mm}
        \subsection{Quantum Circuit Mapping}\label{sec:mapping}
        Although the topic of quantum computing is vast and fascinating, understanding its intricacies is not
        required in order to understand the contents of this work. All one needs to know about quantum computing is
        that, as in classical computing, there is a circuit model for quantum computing. Graphically, a quantum circuit
        is made up of \emph{wires} representing qubits and \emph{gates} representing transformations of qubits denoted
        on the respective wires. A quantum circuit can then be written as a sequence of gates $G = g_0\cdots g_{|G|-1}$,
        where $|G|$ denotes the number of gates in the circuit (also called the circuit's size). Quantum gates can act on
        individual qubits or on multiple qubits simultaneously. 

        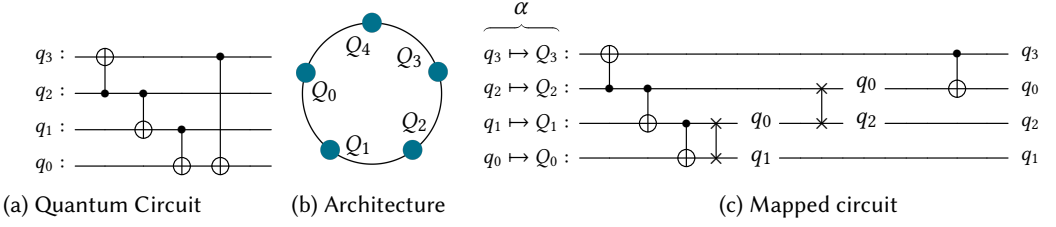
\begin{figure}[t]
          \begin{subfigure}[b]{0.2\linewidth}
            \centering
            \scalebox{0.8}{
              \Qcircuit @C=1.0em @R=0.8em @!R { \\
                \nghost{{q}_{3} :  } & \lstick{{q}_{3} :  } & \targ & \qw & \qw & \ctrl{3} & \qw & \qw\\
                \nghost{{q}_{2} :  } & \lstick{{q}_{2} :  } & \ctrl{-1} & \ctrl{1} & \qw & \qw & \qw & \qw\\
                \nghost{{q}_{1} :  } & \lstick{{q}_{1} :  } & \qw & \targ & \ctrl{1} & \qw & \qw & \qw\\
                \nghost{{q}_{0} :  } & \lstick{{q}_{0} :  } & \qw & \qw & \targ & \targ & \qw & \qw}}
            \caption{Quantum Circuit}\label{fig:ex-circuit}            
          \end{subfigure}\hspace*{.7cm}
          \begin{subfigure}[b]{0.2\linewidth}
            \centering
            \begin{tikzpicture}[overlay]
              \node[] at (0.5,1.1) {\small $Q_0$};
              \node[] at (0.95,0.4) {\small $Q_1$};
              \node[] at (1.7,.7) {\small $Q_2$};
              \node[] at (1.63,1.55) {\small $Q_3$};
              \node[] at (.95,1.7) {\small $Q_4$};
            \end{tikzpicture}
            \includegraphics[width=0.75\linewidth]{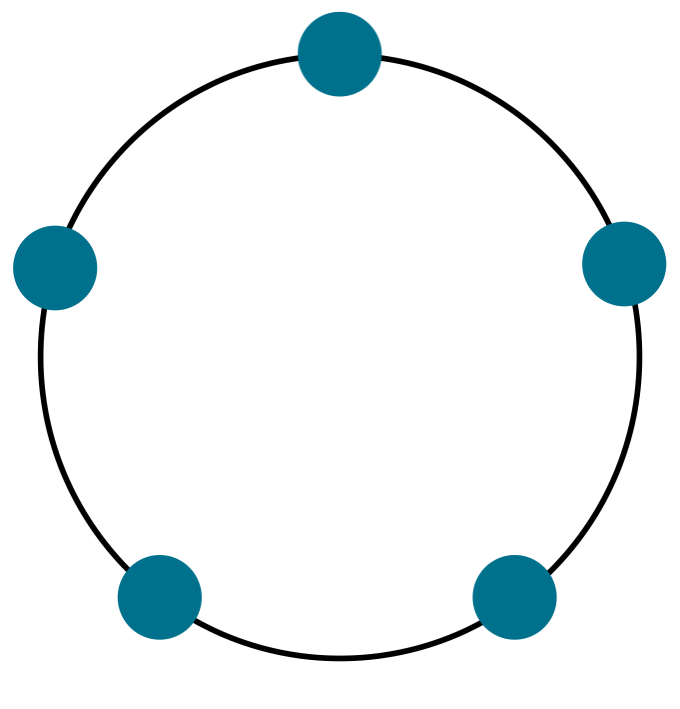}
            \caption{Architecture}\label{fig:arch}
          \end{subfigure}\hspace*{1cm}
          \begin{subfigure}[b]{.49\linewidth}
          \centering
          \begin{tikzpicture}[remember picture, overlay]
            \node[] at (-.5, 2.2) {$\alpha$};
            \draw [decorate, decoration = {calligraphic brace}] (-1,1.9) --  (-0,1.9);
            \node[] at (2.7, 0.23) {\small$q_1$};
            \node[] at (2.7, 0.72) {\small$q_0$};
            \node[] at (4.1, 0.72) {\small$q_2$};
            \node[] at (4.1, 1.2) {\small$q_0$};
          \end{tikzpicture}
          \begin{tikzpicture}
            \node[] at (0, 0) {
              \scalebox{0.8}{
                \Qcircuit @C=1.0em @R=0.8em @!R { \\
                   \lstick{{q}_{3} \mapsto {Q}_{3} :  } & \targ     & \qw      & \qw      & \qw            & \qw & \qw & \qw & \qw & \qw            & \qw & \qw & \qw & \qw & \qw & \ctrl{1} & \qw & \qw   & & \lstick{  {q}_{3}  } \\
                   \lstick{{q}_{2} \mapsto {Q}_{2} :  } & \ctrl{-1} & \ctrl{1} & \qw      & \qw            & \qw & \qw & \qw & \qw & \qswap \qwx[1] & \qw &     &     & \qw & \qw & \targ    & \qw & \qw   & & \lstick{  {q}_{0}  } \\
                   \lstick{{q}_{1} \mapsto {Q}_{1} :  } & \qw       & \targ    & \ctrl{1} & \qswap \qwx[1] & \qw &     &     & \qw & \qswap         & \qw &     &     & \qw & \qw & \qw      & \qw & \qw   & & \lstick{  {q}_{2}  } \\
                   \lstick{{q}_{0} \mapsto {Q}_{0} :  } & \qw       & \qw      & \targ    & \qswap         & \qw &     &     & \qw & \qw            & \qw & \qw & \qw & \qw & \qw & \qw      & \qw & \qw   & & \lstick{  {q}_{1} }
                }} };
          \end{tikzpicture}
	\caption{Mapped circuit}\label{fig:mapping}
      \end{subfigure}        \vspace*{-5mm}
          \caption{Example instance of the mapping problem}\label{fig:mapping-instance}\vspace*{-5mm}
        \end{figure}

        Existing quantum
        computers usually offer a native gate set comprised of a family of singe-qubit gates complemented with some two-qubit gate (e.g. the
        controlled not or $\cnot$ gate). In addition, these quantum computers often do not allow execution of two-qubit gates between any qubit
      pair. Instead, the target device's architecture only allows interactions between certain qubits. This restricted
      connectivity is expressed through the device's \emph{coupling graph} $A = (Q, E)$, where $Q$ are \emph{physical
        qubits} and $(Q_i, Q_j) \in E$ if a two-qubit gate can be executed between $Q_i$ and $Q_j$. Coupling graphs can
      be \emph{directed} or \emph{undirected}. In the former case, executing a gate like the two-qubit $\cnot$ on a directed edge
      $(Q_i, Q_j)$ can only be performed with $Q_i$ as the first and $Q_j$ as the second qubit. In the following, we
      only consider undirected coupling graphs as our ideas can be straightforwardly extended to the directed case. For
      the sake of brevity we will use the term \emph{architecture} when referring to the coupling graph of a quantum
      computing architecture.
	
      To run a quantum algorithm on a target device, the logical qubits~$q$ first have to be mapped to the physical qubits~$Q$. 
      A \emph{qubit assignment} is an injective function $\alpha: q \rightarrow Q$, i.e., a function that
      uniquely assigns each logical qubit a physical qubit. As long as $|q| \leq |Q|$, such a
      mapping can always be obtained. Executing a quantum circuit $G = g_0\cdots g_{|G|-1}$ on a device with coupling graph
      $A$ and initial assignment $\alpha$ is only possible if all two-qubit gates of the circuit act on qubits connected on the
      architecture. If this is not possible, the assignment has to be changed dynamically throughout the circuit in order for each gate
      of $G$ to be executable. This is generally accomplished by adding $\swap$ gates to the circuit, which exchange two qubits in an assignment.

      The mapping problem for a quantum circuit is to find an initial assignment $\alpha$ and a (potentially empty) sequence of
      $\swap$ insertions such that the entire circuit can be executed on the target architecture. 
      We denote the optimal, i.e. minimal, number of swaps when mapping a quantum circuit $G$ to architecture $A$ with $\sopt(G, A)$.
	
        \vspace*{-1mm}
	\begin{example}\label{ex:trivial_mapping}
          Assume the circuit shown in \autoref{fig:ex-circuit} is to be mapped to the architecture defined by the
          coupling graph shown in \autoref{fig:arch}. No initial assignment exists such that this circuit can be executed
          without inserting $\swap$ gates. Taking the initial assignment $\alpha$ depicted on the left-hand side of
          \autoref{fig:mapping} allows the direct execution of the first three $\cnot$ gates. However, the $\cnot$
          between logical qubits $q_3$ and $q_0$ cannot be executed with this assignment. Inserting a $\swap$ between
          physical qubits $Q_1$ and $Q_0$ followed by a $\swap$ between physical qubit $Q_1$ and $Q_2$ leads to the
          assignment $$\{q_3 \mapsto Q_3, q_2 \mapsto Q_1, q_1 \mapsto Q_0, q_0 \mapsto Q_2\}.$$ Therefore, the $\cnot$
          between $q_3$ and $q_0$ translates to a $CNOT$ between $Q_3$ and $Q_2$ which can be directly executed. 
	\end{example}        \vspace*{-1mm}
	
	It is easy to see, that mapping a circuit can always be done by choosing some initial layout and greedily
        inserting $\swap$ gates. This leads to circuits with many $\swap$ gates---an undesirable property on NISQ
        devices due to the relatively high error rate of two-qubit gates. However, finding good solutions to the mapping problem is a challenging task in general and has even been shown to be
        \mbox{NP-hard}~\cite{boteaComplexityQuantumCircuit2018, siraichiQubitAllocation2018,maslovQuantumCircuitPlacement2007,tanOptimalityStudyExisting2021}. Although these works make slightly different assumptions about the qubit mapping problem, the problem's complexity stays the same overall.

        \vspace*{-1mm}
	\section{Considered Problem}\label{sec:considered-problem}
        \vspace*{-1mm}
        
	Given an $n$-qubit quantum circuit $G$ and an architecture $A$ (with $|A|\geq n$), the search space of the mapping problem is spanned by:
	\begin{enumerate}
        \item \emph{Qubit allocation:} 
          Which subarchitecture $A' $ of $ A$ consisting of at least $n$ qubits shall be considered for the subsequent routing?
        \item \emph{Qubit routing:}
          Assuming that $|A'|$ qubits have been allocated, which of the $|A'|!$ permutations---realizing arbitrary
          transitions between assignments via sequences of SWAP gates---to consider in front of every two-qubit gate?
        \end{enumerate}

        \begin{figure}[t]
          \centering
          \includegraphics[width=0.9\linewidth]{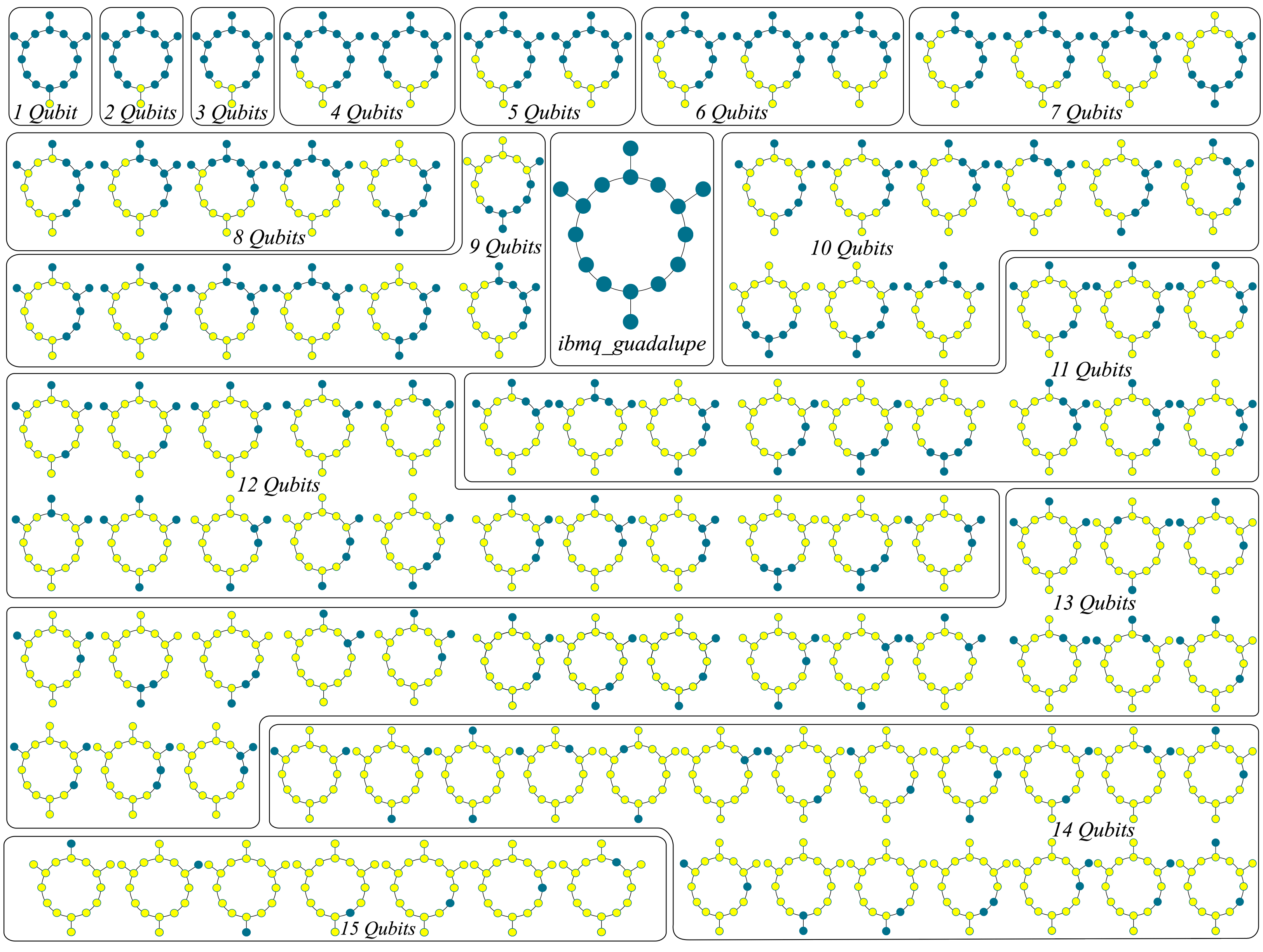}
          \caption[Non-isomorphic subgraphs of \emph{ibmq\_guadalupe}]{Non-isomorphic subgraphs (marked as \tikz{\draw[fill=myyellow,line width=1pt]  circle(1ex);}) of the $16$-qubit \emph{ibmq\_guadalupe} architecture (shown in the middle) grouped by their number of qubits.}
          \label{fig:ibm-guadalupe-subgraphs}\vspace*{-3mm}
        \end{figure}
                      
	As evident from the above, the search space in quantum circuit mapping grows exponentially with respect to the number of physical qubits allocated for the mapping process, i.e., the size of the considered subarchitecture.
	Thus, pruning parts of this search space 
	is absolutely crucial for aiding mapping methods in efficiently traversing the search space.
	However, while limiting the search space, in general, increases the efficiency of mapping techniques (since a smaller search space is easier to explore), it bears the risk of cutting off regions which contain optimal/efficient solutions. 
	Consequently, extra care has to be taken when 
	\begin{enumerate}
		\item reducing the number of allocated physical qubits $|A'|$, or
		\item limiting the considered permutations during qubit routing.
	\end{enumerate}
	Most existing techniques that try to efficiently find good mapping solutions tend to focus on the second aspect, i.e., 
	limiting the number of considered permutations or combinations of SWAPs during qubit
        routing~\cite{liTacklingQubitMapping2019, burgholzerLimitingSearchSpace2022, tanOptimalLayoutSynthesis2020,
          cowtanQubitRoutingProblem2019, sinhaQubitRoutingUsing2022, palerMachineLearningOptimization2022,
          baiolettiNovelAntColony2021, pozziUsingReinforcementLearning2020,debExploringPotentialBenefits2021}.
	On the other hand, the first aspect has hardly been considered thus far.

        \begin{example}\label{ex:subgraphs}
          In order to get an impression of the qubit allocation problem, \autoref{fig:ibm-guadalupe-subgraphs} shows all non-isomorphic subgraphs of the $16$-qubit \emph{ibmq\_guadalupe} architecture---grouped by their number of qubits.
          Assume that a $9$-qubit circuit shall be mapped to this architecture.
          Then, every subgraph with at least $9$ qubits is a potential candidate for mapping---amounting to a
          total of $91$ options.
        \end{example}

        Ideally, one would strive to allocate as few qubits as possible, in particular only as many as there are logical
        qubits in the circuit. This would entail the greatest search-space reduction as only as many qubits are
        allocated as are strictly necessary to ensure the existence of a mapping. Cutting off parts of the search space
        might, however, have the unwanted effect of eliminating optimal or even efficient solutions. Experimental
        evaluations on optimal quantum circuit mapping suggested that this is not the
        case~\cite{willeMappingQuantumCircuits2019,burgholzerLimitingSearchSpace2022}.

        This work remedies this misconception by showing that considering larger than strictly necessary subarchitectures 
        can be beneficial when searching for efficient solutions. This
        shows, once and for all, that considering subarchitectures of minimal size is not sufficient. 
        However, instead of being content with a negative result, we then explore the following problem.

        \begin{problem}        
          Let $A'$ and $A''$ be two subarchitectures of an architecture $A$. $A'$ is said to possess \emph{higher
            coverage} than $A''$, i.e., $A'' <_\text{cov}^n A'$ with respect to $n \in \mathbb{N}$ if:

          \begin{itemize}
          \item Any quantum circuit $G$ over $n$ qubits can be optimally mapped to $A'$ with no more $\swap$s than are needed for
            mapping $G$ to $A''$, i.e., $\forall G: \sopt(G,A') \leq \sopt(G, A'')$.
          \item There is a quantum circuit $G$ over $n$ qubits that can be mapped to $A'$ using less $\swap$s than required for mapping
            it to $A''$,~i.e, $\exists G: \sopt(G, A') < \sopt(G, A'')$.
          \end{itemize}
          If only the first property is fulfilled we write $A'' \leq_\text{cov}^n A'$.
          
          The problem of finding the optimal subarchitectures of $A$ for a given size $n$ is the problem of finding the
          smallest (with respect to the size of the subarchitectures) maximal elements with respect
          to~$<_\text{cov}^n$,~i.e., $$\min_{|A'|} \; \max_{\leq_\text{cov}^n} \{A' : A' \sqsubseteq A\}.$$ 
        \end{problem}

\section{Existence of Better Mappings on Larger Subarchitectures}\label{sec:better-mappings}

In this section, we show that considering larger subgraphs can, under certain circumstances, lead to better solutions
for the mapping problem.
To this end, \autoref{sec:shorter_connection} and \autoref{sec:more_subgraphs} establish two criteria that allow
for ranking subarchitectures according to their coverage. These criteria are constructive in the sense that
they allow for the definition of a method for computing good subarchitectures (with respect to $<_\text{cov}^n$)---trading off subarchitecture size and, by extension, search space size for the mapping
problem, and likelihood of eliminating efficient mapping solutions. Afterwards, we give a complete
characterization of the smallest maximal elements with respect to $<_\text{cov}^n$ among all the non-isomorphic
subarchitectures of a given architecture. This proves that merely considering the subgraphs $A'$ of an
architecture $A$ with as many qubits as the 
circuit to be mapped can lead to situations where the optimal solution to the mapping problem is no longer contained in
the search space induced by the allocated qubits.

The proofs in this work require somewhat technical circuit constructions. Instead of listing the details here, an intuition of these constructions will be given. These can be generalized to arbitrary (sub-)architectures.

\vspace*{-1mm}
\subsection{Subarchitectures with Shorter Connections}\label{sec:shorter_connection}
\vspace*{-1mm}
The first criterion is based on the observation that subarchitectures with shorter connections between certain qubits potentially allow for better mappings---even if they involve more qubits.
 An example illustrates the idea:

        \begin{figure}[t]
          \centering
                    \begin{subfigure}[b]{0.2\linewidth}
            \centering
            \scalebox{0.8}{
              \Qcircuit @C=1.0em @R=0.8em @!R { \\
                \nghost{{q}_{3} :  } & \lstick{{q}_{3} :  } & \targ & \qw & \qw & \ctrl{3} & \qw & \qw\\
                \nghost{{q}_{2} :  } & \lstick{{q}_{2} :  } & \ctrl{-1} & \ctrl{1} & \qw & \qw & \qw & \qw\\
                \nghost{{q}_{1} :  } & \lstick{{q}_{1} :  } & \qw & \targ & \ctrl{1} & \qw & \qw & \qw\\
                \nghost{{q}_{0} :  } & \lstick{{q}_{0} :  } & \qw & \qw & \targ & \targ & \qw & \qw}}
            \caption{Quantum Circuit}
          \end{subfigure}\hspace*{.7cm}
          \begin{subfigure}[b]{0.2\linewidth}
            \centering
            \begin{tikzpicture}[overlay]
              \node[] at (0.5,1.1) {\small $Q_0$};
              \node[] at (0.95,0.4) {\small $Q_1$};
              \node[] at (1.7,.7) {\small $Q_2$};
              \node[] at (1.63,1.55) {\small $Q_3$};
              \node[] at (.95,1.7) {\small $Q_4$};
            \end{tikzpicture}
            \includegraphics[width=0.75\linewidth]{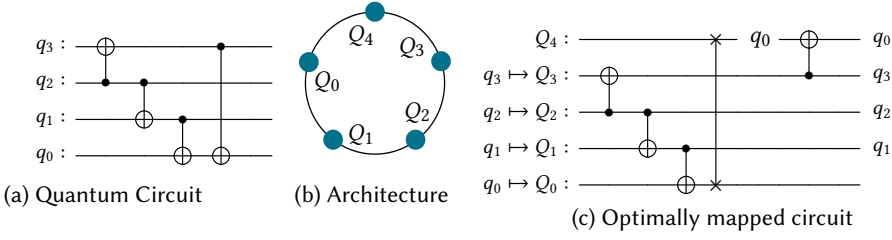}
            \caption{Architecture}
          \end{subfigure}
          \begin{subfigure}[b]{.44\linewidth}           
          \begin{tikzpicture}[overlay]
            \node[] at (3.7,2.15) {\small$q_0$};
          \end{tikzpicture}
          \begin{tikzpicture}
            \node[]        {\scalebox{0.8}{
                \Qcircuit @C=1.0em @R=0.8em @!R { \\
                  \nghost{{Q}_{0} :  } & \lstick{{Q}_{4} :  }                 & \qw       & \qw      & \qw      & \qswap \qwx[4] & \qw &     &     & \targ     & \qw & \qw &  &  \lstick{  {q}_{0}} \\
                  \nghost{{Q}_{1} :  } & \lstick{{q}_{3} \mapsto {Q}_{3} :  } & \targ     & \qw      & \qw      & \qw            & \qw & \qw & \qw & \ctrl{-1} & \qw & \qw &  &  \lstick{  {q}_{3}} \\
                  \nghost{{Q}_{2} :  } & \lstick{{q}_{2} \mapsto {Q}_{2} :  } & \ctrl{-1} & \ctrl{1} & \qw      & \qw            & \qw & \qw & \qw & \qw       & \qw & \qw &  &  \lstick{  {q}_{2}} \\
                  \nghost{{Q}_{3} :  } & \lstick{{q}_{1} \mapsto {Q}_{1} :  } & \qw       & \targ    & \ctrl{1} & \qw            & \qw & \qw & \qw & \qw       & \qw & \qw &  &  \lstick{  {q}_{1}} \\
                  \nghost{{Q}_{4} :  } & \lstick{{q}_{0} \mapsto {Q}_{0} :  } & \qw       & \qw      & \targ    & \qswap         & \qw & \qw & \qw & \qw       & \qw & \qw &  &  \lstick{ }
                }}            };
          \end{tikzpicture}\vspace*{-2mm}
          \caption{Optimally mapped circuit}\vspace*{-3mm}
          \end{subfigure}
          \caption{Optimal solution to the mapping problem from \autoref{fig:mapping-instance} considering the complete architecture}
          \label{fig:circuit-mapped-cycle}\vspace*{-3mm}
        \end{figure}
        
        \begin{example}\label{ex:circle}
          Consider again the circuit shown in \autoref{fig:ex-circuit} and assume it shall be mapped to
          the $5$-qubit ring architecture shown in \autoref{fig:arch}. 
          Previously, in \autoref{ex:trivial_mapping}, the circuit has been mapped to the $4$-qubit line
          $Q_0 - Q_1 - Q_2 - Q_3$, which resulted in the addition of two SWAP gates in order to satisfy
          the coupling constraints (as illustrated in \autoref{fig:mapping}).	 
          It can be shown that this is the optimal number of SWAP operations given that only four of the architecture's qubits are considered.

          If, however, the whole architecture is considered (i.e., including the idle qubit $Q_4$), the optimal solution just requires a single SWAP.
          The resulting circuit is shown in \autoref{fig:circuit-mapped-cycle}.
        \end{example}

        Based on the previous example, it is clear that extending a \mbox{subarchitecture} to a larger one might admit a
        better mapping for a quantum circuit if the path between some qubits is shorter in the larger subgraph.  
        This insight suggests a natural way of ordering the subarchitectures of an architecture:

        \begin{definition}\label{def:order}
        		Let the (partial) order $\preceq$ be defined as a relation on graphs with the following properties:\\
           	Let \mbox{$A = (V_A, E_A)$} and \mbox{$A' = (V_{A'}, E_{A'})$} be two graphs. 
           	Then, $A' \prec A$ iff $A' \sqsubseteq A$ with subgraph isomorphism $h$ and there are two nodes that have a shorter path between them in $A$ than the corresponding nodes in $A'$, i.e.,   \[\exists\; v, w \in V_{A'}\colon |p(v,w)| > |p(h(v), h(w))|.\]
                If $A'$ and $A$ are only subgraph isomorphic with no shorter paths then we write $A' \preceq A$.                
        \end{definition}

        This order has the desirable property that $A' \preceq A$ iff the optimal mapping of an arbitrary quantum circuit~$G$ to the larger architecture~$A$ is
        at least as good as the mapping of~$G$ to subarchitecture~$A'$, i.e., $A' \leq_\text{cov}^{|A'|} A$.
        The reason for this is simple: Since $A'$ is a subarchitecture of $A$, any quantum circuit $G$ mapped to $A'$ can be mapped to $A$ in exactly the same fashion.

        Moreover, this partial ordering provides a way of ranking the subarchitectures of an architecture:
        \begin{theorem}\label{thm:opt-layout}
          Let $A'$ be a subarchitecture of an architecture $A$ such that $A' \prec A$.
          Then, a quantum circuit $G$ exists that is cheaper to map to $A$ instead of $A'$, i.e., $\sopt(G, A) <
          \sopt(G, A')$ and, by extension, $A' <_\text{cov}^{|A'|} A$.
        \end{theorem}

        \begin{proof}
          The intuition behind the proof is that, given $A'\prec A$, there exist vertices $v$ and $w$ in $A'$ that have a shorter connection in $A$.
          If, during mapping, an interaction between $v$ and $w$ is required, this interaction can be realized more efficiently in $A$ due to the shorter connection.
        \end{proof}

        \begin{example}
          Consider again the scenario from \autoref{ex:circle}. 
          Then, the $4$-qubit line is strictly smaller than the $5$-qubit ring with respect to $\preceq$ (i.e., $\mathtt{line} \prec \mathtt{ring}$)
          since it holds that 
          \[|p_\mathit{line}(Q_0, Q_3)| = 3 > 2 = |p_\mathit{ring}(Q_0, Q_3)|.\]
          As a consequence of \autoref{thm:opt-layout}, there exists a circuit (e.g., the one shown in \autoref{fig:ex-circuit}) which has a more efficient mapping solution on the $5$-qubit ring (see \autoref{fig:circuit-mapped-cycle}) as compared to the $4$-qubit line (see \autoref{fig:mapping}).
        \end{example}

        	For any graph $A$ and subgraph $A'$, the set of \emph{desirable} subarchitectures $\mathcal{D}(A', A)$ is given by the maximal elements with respect to $\prec$, i.e., 
        	\[
        	\mathcal{D}(A', A) = \left\{D\colon\; A'\sqsubseteq D\sqsubseteq A \mbox{ and } \not\exists\,D'\sqsubseteq A\colon\; A' \prec D' \wedge D \prec D' \right\}.
        	\]
        	This set contains all subarchitectures worth considering during mapping due to their potential of providing a better mapping solution according to the ranking defined by the order $\preceq$.

                \vspace*{1mm}
                \subsection{Subarchitectures without Shorter Connections}\label{sec:more_subgraphs}
                \vspace*{2mm}
		
		As shown above, a better mapping can potentially be achieved if two qubits have a shorter connection on a larger subarchitecture than the one originally considered.
		The natural follow-up question is: Can a single additional qubit, that does \emph{not} lead to a shorter connection between some qubits, bring any improvements in mapping a particular quantum circuit?
		In the following, we give an affirmative answer.
		The main idea is that, given an architecture $A$ and a subarchitecture $A'$, a larger subarchitecture
                $A'\sqsubseteq A'' \sqsubseteq A$ can allow for a better mapping, if it contains at least one more
                subarchitecture of size $|A'|$ that is not isomorphic to $A'$. 
		Again, an example illustrates the idea:

        \begin{example}\label{ex:iso-opt}
          Instead of a $5$-qubit ring (as in \autoref{ex:circle}), consider an architecture as shown on the bottom of \autoref{fig:iso-opt}.
          This architecture contains two non-isomorphic $4$-qubit subarchitectures as shown in the middle of \autoref{fig:iso-opt}---a \enquote{line} and a \enquote{T}-shaped subarchitecture.
          The quantum circuits $G$ and $G'$ shown on top of the subarchitectures can be mapped to the respective subarchitecture
          without SWAPs, while they require a single SWAP to be realized on the other subarchitecture.
          	
          Now, consider the quantum circuit $G^4 G^{\prime 4}$, i.e., a circuit composed of four repetitions of each of the two circuits described above.
          Mapping this circuit to either of the $4$-qubit subarchitecture requires at least four SWAPs
          (since every repetition of the circuit which was not designed for the respective
          subarchitecture introduces a SWAP). 
          However, it only takes three SWAPs to transform the qubit assignment on the \enquote{line} subarchitecture to
          a qubit assignment on the \mbox{\enquote{T}-shaped} subarchitecture on the whole architecture.
          Therefore, considering the whole architecture allows for a better mapping than mapping to either \mbox{subarchitecture}.
        \end{example}

        \begin{figure}[t]
          	\centering
          	\includegraphics[width=0.5\linewidth]{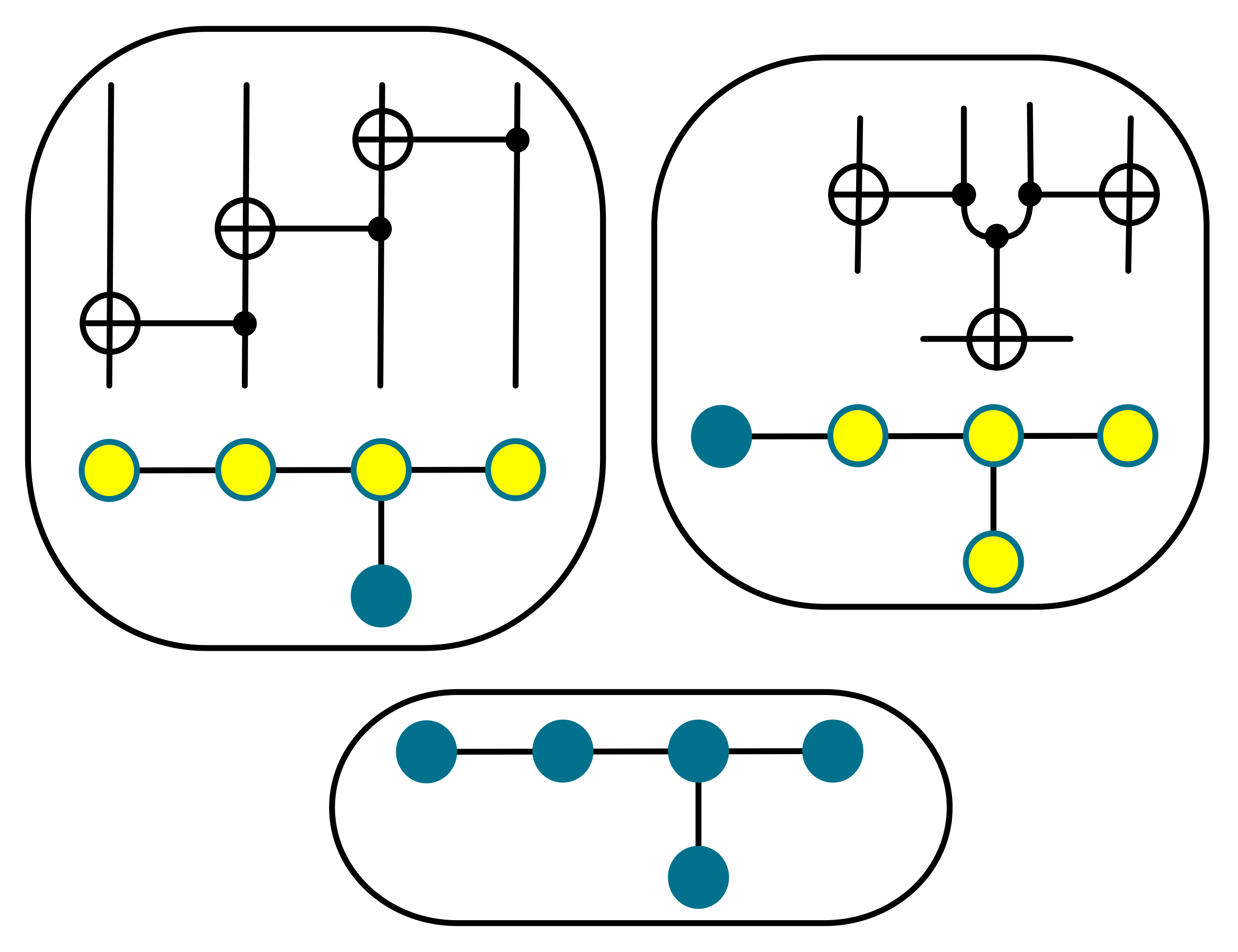}\vspace*{-3mm}
          \caption[Two layout scenario]{$5$-qubit architecture (bottom), its $4$-qubit subgraph candidates (marked as \tikz{\draw[fill=myyellow,line width=1pt]  circle(1ex);}), and two circuits each of which can be mapped to one of the subgraphs without SWAPs (top).}
          \label{fig:iso-opt}
        \end{figure}
        
        	The following theorem formalizes this observation:

        \begin{theorem}\label{thm:iso-opt}
        		Let $A', A''$ be equally-sized (i.e., \mbox{$|A'| = |A''| = n$}), non-isomorphic (i.e., $A'\not\cong A''$)
                        proper induced subarchitectures of an architecture $A$ (i.e., $A'\sqsubset A$ and $A''\sqsubset A$).%
        		Then there is a quantum circuit $G$ that is cheaper to map to $A$ than to each individual
                        \mbox{subarchitecture}, i.e., $A' <_\text{cov}^n A$ and $A'' <_\text{cov}^n A$. 
        \end{theorem}
        
        \begin{proof}
          The intuition behind this proof is that, similar to \autoref{ex:iso-opt}, for each subarchitecture, a circuit can be
          constructed, that does not require any SWAPs to be executed on that subarchitecture, but will require at
          least one SWAP on the other subarchitecture. 
          In addition, there is a certain number of SWAPs~$S$ that is needed to transform an assignment to $A'$ to an
          assignment to $A''$ on the whole architecture.  
          If $G'$ ($G''$) is the circuit corresponding to $A'$ ($A''$), then the circuit $G = (G')^{S+1} (G'')^{S+1}$ is
          cheaper to map to $A$ than to either $A'$ or $A''$. 
        \end{proof}

        Thus, even a single additional qubit that does not induce a shorter connection may yield a better overall mapping.

        \subsection{Characterizing Optimal Subarchitectures}\label{sec:opt}
        In the above, the benefit of an additional qubit is due to different parts of a circuit being more efficiently executable on different \mbox{subarchitectures}.
        This implies that all the desirable \mbox{subarchitectures} need to be covered in order to not lose optimality.
        The following example demonstrates, that extra care needs to be taken when determining such a covering.

                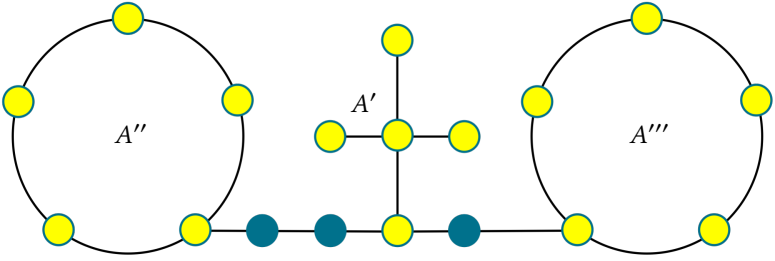
\begin{figure}[t]
                  \centering
                  \begin{tikzpicture}[overlay]
                    \node[] at (1.9,1.7) {$A''$};
                    \node[] at (5,2.1) {$A'$};
                    \node[] at (8.8,1.7) {$A'''$};
                  \end{tikzpicture}
                  \includegraphics[width=0.75\linewidth]{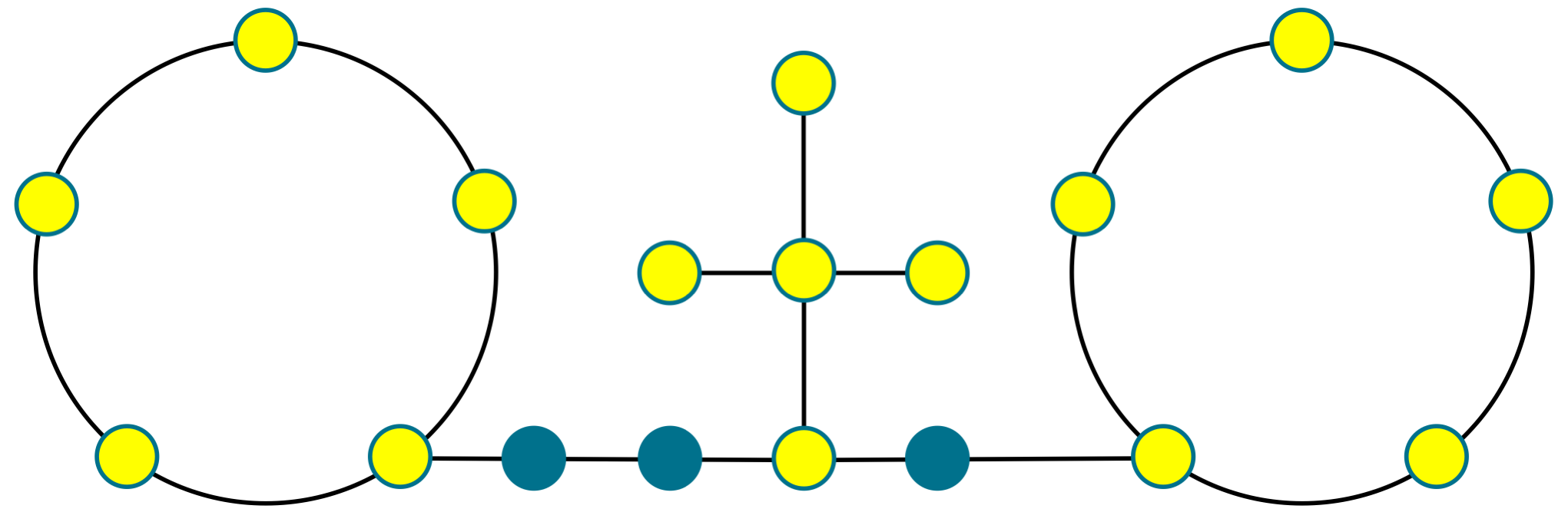}
                  \caption{Subarchitecture transformations}
                  \label{fig:transformation}
                \end{figure}

                \begin{example}
                  The architecture $A$ shown in \autoref{fig:transformation} has, among others, the three $5$-qubit subarchitectures labeled $A', A''$ and $A'''$.
                	It is immediate that $A''$ and $A'''$ are isomorphic. 
                	According to \autoref{thm:iso-opt}, there exists a $5$-qubit circuit $G$ that can be more efficiently mapped to a subarchitecture of $A$ containing $A'$ and $A''$ than to each individual subarchitecture.
                	There are two possible options for covering both subgraphs on $A$ by either connecting $A'-A''$ or $A'-A'''$.
                	However, since the connection between $A'$ and $A''$ is shorter, it has to be ensured that this graph is chosen over the alternative.

                \end{example}
            
                Thus, a truly-optimal subarchitecture not only needs to cover any of the desirable subarchitectures
                but also has to allow for the shortest possible transformation between them. 
                The last piece of the puzzle to characterize optimal subarchitectures of an architecture is that merely
                considering subarchitectures still does not suffice to guarantee that the optimal solution is not
                lost.

                \begin{example}
                	Consider again the circuit in \autoref{fig:ex-circuit}. \autoref{ex:circle} showed that this circuit can be more efficiently mapped to a $5$-qubit ring than to a $4$-qubit line.
                	In \autoref{sec:shorter_connection}, this was attributed to the shorter connection between two qubits in the larger architecture.
                	A more general observation is that any arrangement of (logical) qubits on the device can potentially be permuted more efficiently on the larger architecture.
                	Similarly, in \autoref{ex:iso-opt}, the $5$-qubit architecture allows for the transformation of
                        any linear arrangement of qubits to a \mbox{\enquote{T}-shaped} arrangement, which was not possible on
                        either of it subgraphs. 

                \end{example}

                This brings us to the final characterization and theorem for optimal subarchitectures.
                Unfortunately, the precise formulation and proof of this theorem are very technical.
                For the sake of brevity, an informal version of the theorem is listed here.

                \begin{theorem}
                  \label{thm:opt-informal}
                  The $n$-qubit optimal subarchitectures· of an architecture are those that
                  allow to realize any sequence of $n$-qubit arrangements as efficiently as it could be realized on the
                  whole architecture. 
                \end{theorem}

                This is an extremely restrictive requirement! There are way more arrangements of $n$ qubits than there
                are $n$-qubit subarchitectures of an architecture. For example, while there is only one $4$-qubit
                \enquote{line} subarchitecture up to isomorphism, there are twelve \mbox{non-unique} ways to arrange four qubits on a
                line. 
                Any sequence of these arrangements needs to be realizable as efficiently on the potential subarchitecture than it would be on the whole architecture.
                This makes the \enquote{hunt} for optimal subarchitectures of an architecture look grim. It seems that not much
                is gained by characterizing optimal subarchitectures according to \autoref{thm:opt-informal} as it
                still does not suggest a
                method for computing optimal subarchitectures. The real insight of this theorem is that it demonstrates just how
                general optimal subarchitectures are and that qubits cannot easily be neglected for the purpose of
                reducing the search space in quantum circuit mapping.
                
                One might even be tempted to conclude that on any given architecture $A$ one cannot neglect \emph{any}
                physical qubits when mapping a quantum circuit $G$ over $n < |A|$ qubits without potentially
                excluding essential parts of the search space.
                However, this pessimism is misplaced, as the following simple example shows.

                \begin{example}
                  Consider a $5$-qubit \enquote{line} architecture. Then, the optimal mapping of any $4$-qubit quantum
                  circuit requires only four physical qubits. \autoref{thm:opt-informal} explains why: 
                  None of the twelve possible arrangements of four qubits on a line can be transformed into each other more efficiently using the additional qubit.
                \end{example}

                Computing optimal subarchitectures according to \autoref{thm:opt-informal} is an incredibly complex
                problem due to the sheer number of possible arrangements. Fortunately, the requirements of
                \autoref{thm:opt-informal} are way to strict in practice and it generally suffices to cover only the
                subarchitectures suggested by \autoref{thm:opt-layout} and/or	\autoref{thm:iso-opt}. 
                How these near-optimal sets of subarchitectures are computed, is shown next.
                
		\section{Computing near-optimal Subarchitectures}\label{sec:methods}
		
		Based on the criteria derived in the previous section, we now demonstrate two ways of computing sets of
                subarchitectures to be considered for mapping which only depend on the number of qubits of the circuit to be
                mapped. 
		As a result, for any particular architecture, these \emph{subarchitecture candidates} can be computed a-priori without any particular information about the circuits.

		\subsection{Candidates for Optimal Subarchitectures}\label{sec:opt-graphs}
		
		Assume that an $n$-qubit circuit $G$ is to be mapped to an architecture $A$ with $|A|\geq n$.
		Let $\mathit{Sub}_n(A)$ denote the set of all induced subarchitectures $A'$ of $A$ with size $n$, i.e., 
		\[\mathit{Sub}_n(A) = \{ A'\colon\; A' \sqsubseteq A \mbox{ and } |A'| = n \}. \]
		
		According to \autoref{thm:opt-layout}, this set of subarchitectures is
                probably not a set of optimal subarchitectures since the architecture might contain \enquote{better} subarchitectures. 
		Instead, for each of its elements, the corresponding set of \emph{desirable} subarchitectures should be considered.
		Thus, an initial set of subarchitecture candidates $\cand_n(A)$ is given by
		\[
			\cand_n(A) = \bigcup_{A' \in \mathit{Sub}_n(A)} \mathcal{D}(A', A).
		\]

        \begin{figure}[t]
          \centering
          \begin{minipage}[b]{.59\linewidth}
          \includegraphics[width=\linewidth]{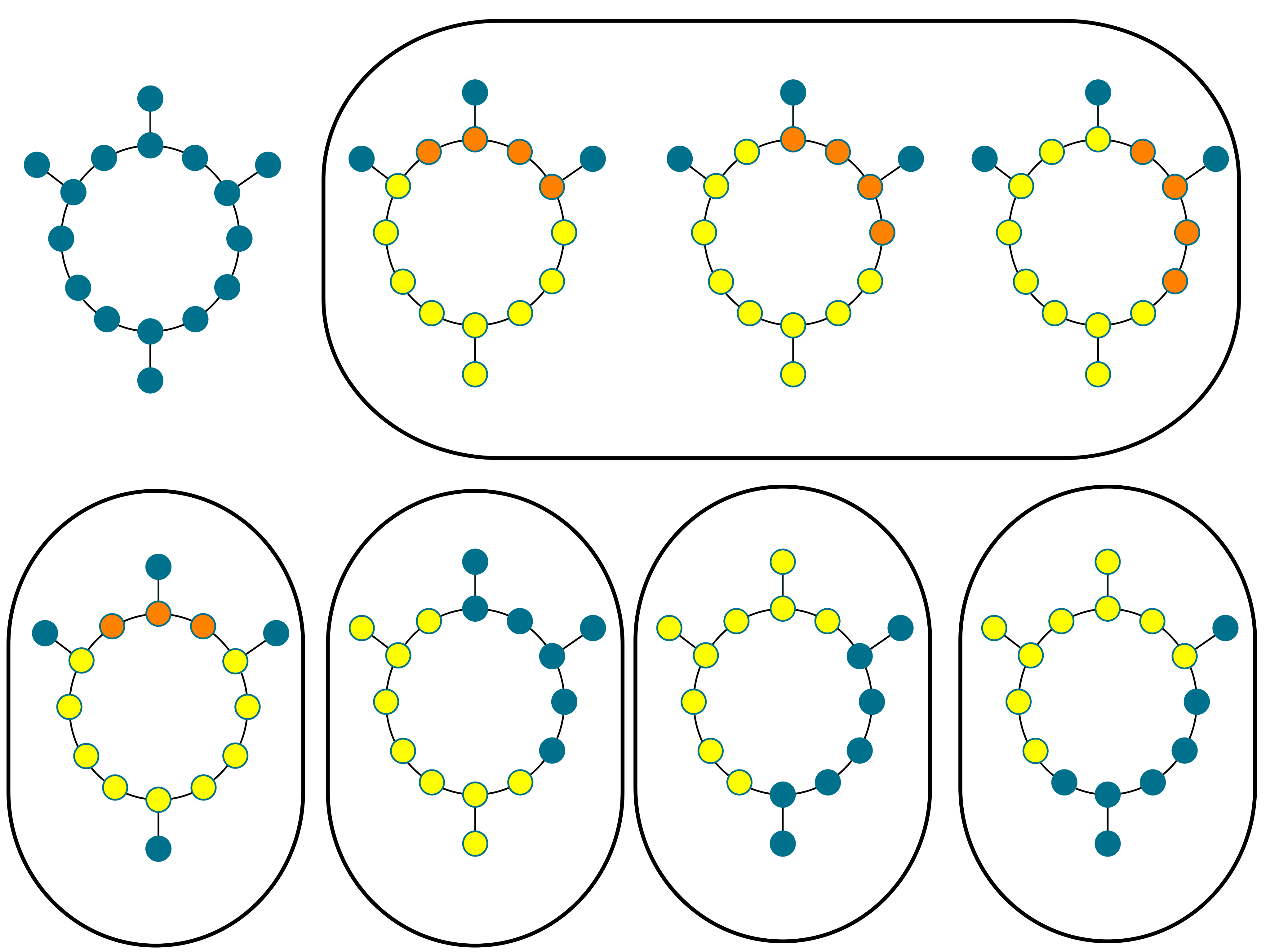}
          \caption[Subarchitecture candidates]{\emph{ibmq\_guadalupe} architecture (top left), its non-isomorphic $9$-qubit subarchitectures (marked as \tikz{\draw[fill=myyellow,line width=1pt]  circle(1ex);}) and the corresponding desirable subarchitectures (with additional qubits marked as \tikz{\draw[fill=myorange,line width=1pt]  circle(1ex);} and qubits not included in the subarchitecture marked as \tikz{\draw[fill=myblue,line width=1pt]  circle(1ex);})---forming the set of subarchitecture candidates for $9$-qubit circuits. Multiple subarchitectures may lead to the same candidate.}
          \label{fig:ibm-guadalupe-candidates}
          \end{minipage}\hspace*{3mm}
          \begin{minipage}[b]{.39\linewidth}
          \centering
          \includegraphics[width=\linewidth]{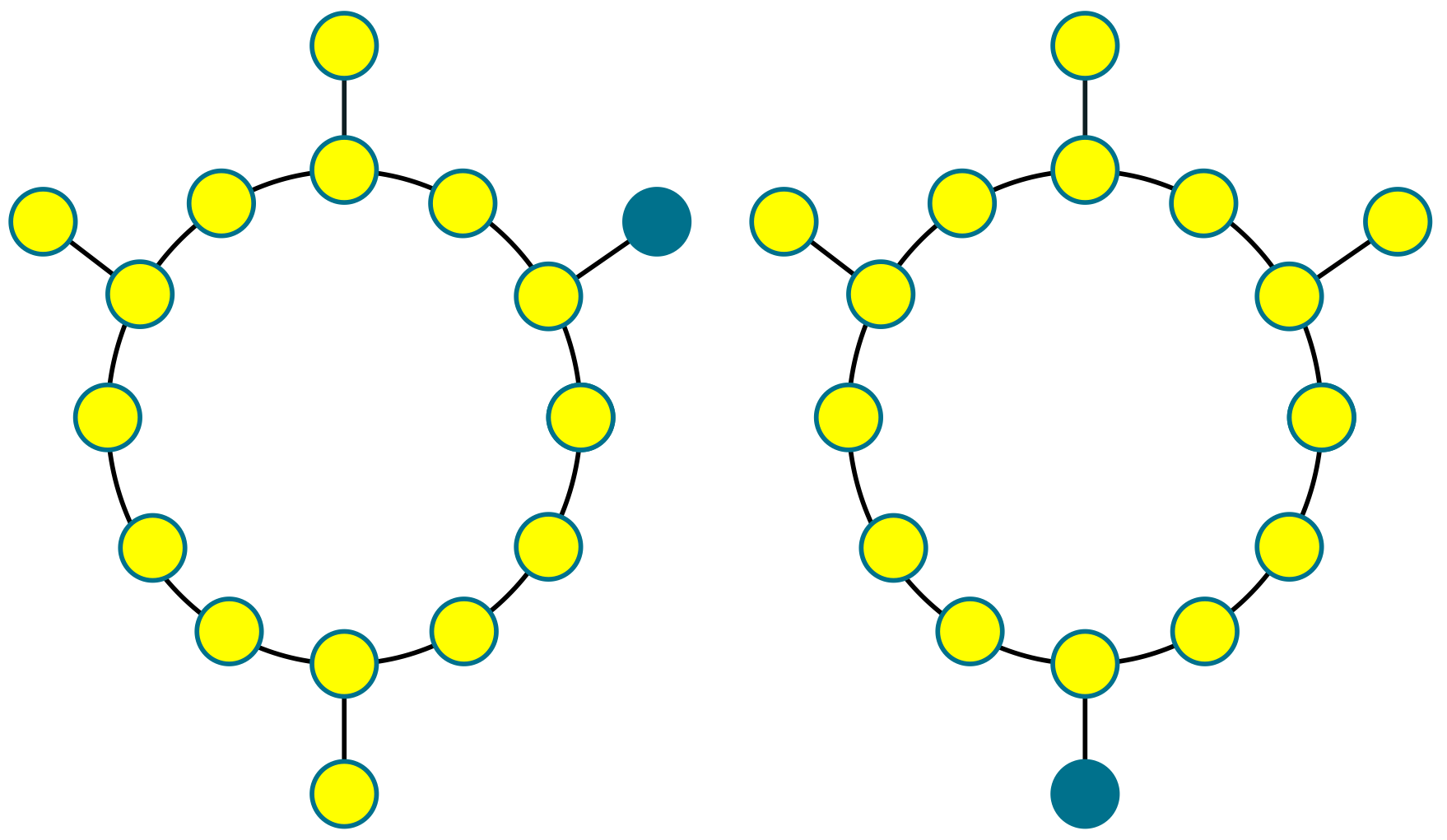}\vspace*{1.5cm}
          \caption[Optimal subgraph candidates]{Optimal subarchitecture candidates (marked as \tikz{\draw[fill=myyellow,line width=1pt]  circle(1ex);}) for $9$-qubit circuits on the \emph{ibmq\_guadalupe} architecture.}
          \label{fig:ibm-guadalupe-minlayout}
        \end{minipage}
        \end{figure}
                
        \begin{example}
          Consider the $16$-qubit \emph{ibmq\_guadalupe} architecture shown on
          the top left of \autoref{fig:ibm-guadalupe-candidates} and assume that a $9$-qubit circuit shall be mapped to
          this architecture. 
          Then, \autoref{fig:ibm-guadalupe-candidates} shows the architecture's seven non-isomorphic $9$-qubit subarchitectures (marked as \tikz{\draw[fill=myyellow,line width=1pt]  circle(1ex);}).
          On this architecture, the only way to decrease the distance between two qubits is to close the central ring of the architecture.
          In four out of seven cases, closing the ring indeed shortens the path between two qubits (marked as \tikz{\draw[fill=myorange,line width=1pt]  circle(1ex);}).
          Overall, the set of mapping candidates for $9$-qubit circuits $\cand_9(\mathit{ibmq\_quadalupe})$ consists of five subarchitectures ranging from nine to thirteen qubits.			
        \end{example}

        According to \autoref{thm:iso-opt}, if a \mbox{subarchitecture} subsumes two
        non-isomorphic smaller subarchitectures, then this larger subarchitecture needs to be considered in order to not lose
        out on optimality. 
        Thus, the set of \emph{optimal} subarchitecture candidates $\mathit{OptCand}_n(A)$ is given by the smallest
        subarchitectures of $A$ (with respect to $\sqsubseteq$) that contain \emph{all} subarchitecture candidates for $n$
        qubits, i.e.,   
        \[
          \mathit{OptCand}_n(A) = \min_{\sqsubseteq} \left( \underbrace{\bigcap_{C \in \cand_n(A)} \overbrace{\{A'\colon\; C \sqsubseteq A' \sqsubseteq A\}}^{\mbox{\scriptsize subarchitectures of $A$ containing $C$}}}_{\mbox{\scriptsize intersection over all subarchitecture candidates}} \right)
        \]        	\vspace*{2mm}

        This set is never empty, since, in the worst case, it consists of $A$ itself (the entire
        architecture, per definition, has all subarchitecture candidates as subarchitectures). 
        Although, in many cases, the elements of $\mathit{OptCand}_n(A)$ contain less qubits. Note
        that these candidates are optimal in the sense that no
        larger subarchitectures can be constructed using  \autoref{thm:opt-layout} and \autoref{thm:iso-opt} that are an improvement with respect to $\leq_\text{cov}^n$.

        \begin{example}
          Consider again the ibmq\_guadalupe architecture from the previous example.
          Then, \autoref{fig:ibm-guadalupe-minlayout} shows the optimal subarchitecture candidates for $9$-qubit circuits.
          It is easy to check, that both indeed contain all five subarchitecture candidates illustrated in
          \autoref{fig:ibm-guadalupe-candidates} and that removing any of their qubits would exclude one of them. 
          While, in this case, only a single qubit is saved compared to the whole architecture, this still constitutes a
          significant reduction in the search space (e.g., by a factor of $16$ due to only having to consider $15!$
          instead of $16!$ possible permutations). 
        \end{example}

                By utilizing the derived subarchitectures to map quantum circuits, the search space can be greatly reduced without cutting away essential parts.

        \subsection{Adaptively Covering Subarchitecture Candidates}\label{sec:covering}

        Although $\mathit{OptCand}_n(A)$ describes exactly those subarchitectures of an architecture $A$ that are
        optimal with respect to \autoref{thm:opt-layout} and \autoref{thm:iso-opt}, the requirement that these
        subarchitectures must contain \emph{all} subarchitecture candidates is a very strong one.
        Since the resulting subarchitectures need to encompass many \mbox{non-isomorphic} subarchitectures, they frequently are large in size compared to the circuit to be mapped. 
        \autoref{thm:iso-opt} implies that this requirement cannot be dropped without potentially cutting away parts of the search space potentially containing optimal solutions.
        However, the structure of the circuits in the proof of \autoref{thm:iso-opt} is still rather
        \enquote{artificial} since converting between different $n$-qubit subarchitectures in a single circuit is hardly relevant
        in practice. Thus, this criterion might be deliberately dropped in order to reduce the size of the
        \mbox{subarchitectures} to be considered---leaving only the subarchitecture candidates for mapping. 

        The number of subarchitecture candidates $\cand_n(A)$ is, in general, bounded.
        In particular, if a architecture $A$ with $n$ qubits has diameter $d$, at most $\frac{n(n-1)}{2} (d-1)$ qubits can be added to improve the connectivity of the circuit.
        As architectures for currently existing devices exhibit low connectivity, the number of qubits that can be added to improve the connectivity of an architecture tends to be quite small.
        However, the sheer number of candidates can grow rather large and trying out every possible one can become exceedingly expensive---even if taking parallel execution into account.
        Thus, it might be desirable to find a collection of \mbox{subarchitectures} (with a limited number of elements) that covers \emph{all} the \mbox{subarchitecture} candidates.

        \begin{definition}
          Let $A$ be an architecture and let $\cand_n(A)$ denote the set of subarchitecture candidates for $n$-qubit circuits. 
          Then, the set of covering candidates $\mathit{CovCand}_n(A)$ is given by the \mbox{subarchitectures} of $A$ containing any of the subarchitecture candidates, i.e.,
\vspace*{1mm}          \[
            \mathit{CovCand}_n(A) = \bigcup_{C \in \cand_n(A)} \{A'\colon\; C \sqsubseteq A' \sqsubseteq A\}
          \] \vspace*{1mm}
          Moreover, a \emph{subarchitecture covering} $\mathit{Cov}_n(A)$ is a subset of $\mathit{CovCand}_n(A)$ with
          the property that for every subarchitecture candidate, there is at least one element in the covering that
          contains that candidate as a subarchitecture, i.e., $\mathit{Cov}_n(A) \subseteq \mathit{CovCand}_n(A)$ such
          that \vspace*{1mm}
          \[
            \forall\, C\in\cand_n(A)\; \exists\, A' \in \mathit{Cov}_n(A)\colon\; C\sqsubseteq A'.
          \]
        \end{definition}

          \begin{minipage}{.6\linewidth}
          \centering         
                \begin{algorithm}[H]
          \caption{Compute subarchitecture covering}\label{alg:covering}
          \begin{algorithmic}
            \Input
          	\Desc{$A$}{Architecture}
          	\Desc{$n$}{Number of qubits}
          	\Desc{$k$}{Max. size of covering}
          	\EndInput
          	\Output
          	\Desc{$\mathit{Cov}_n(A)$}{Covering of $A$ such that $|\mathit{Cov}_n(A)| \leq k$}
          	\EndOutput\\\hrulefill
            \State Initialize $\cov_n(A) \leftarrow \cand_n(A)$
            \State Initialize $\text{queue} \leftarrow \bigcup_{C \in \cand_n(A)}\{A'\colon\; C \sqsubseteq A' \sqsubseteq A\}$
            \State Sort queue with respect to the number of vertices.
            \While {$|\cov_n(A)| > k$} 
            \State $D \leftarrow \operatorname{pop}(\text{queue})$
            \State Let $\cov_D \leftarrow \{C \in \cov_n(A)\colon\; C \sqsubseteq D\}$
            \If {$|\cov_D| > 1$}
            \State Update $\cov_n(A) \leftarrow (\cov_n(A) \setminus \cov_D) \cup \{D\}$
            \EndIf
            \EndWhile
          \end{algorithmic}
        \end{algorithm}
      \end{minipage}\hspace*{.5cm}
      \begin{minipage}{.36\linewidth}
                \begin{figure}[H]
          \centering
          \begin{subfigure}[b]{\linewidth}
          \centering
            \includegraphics[width=0.5\linewidth]{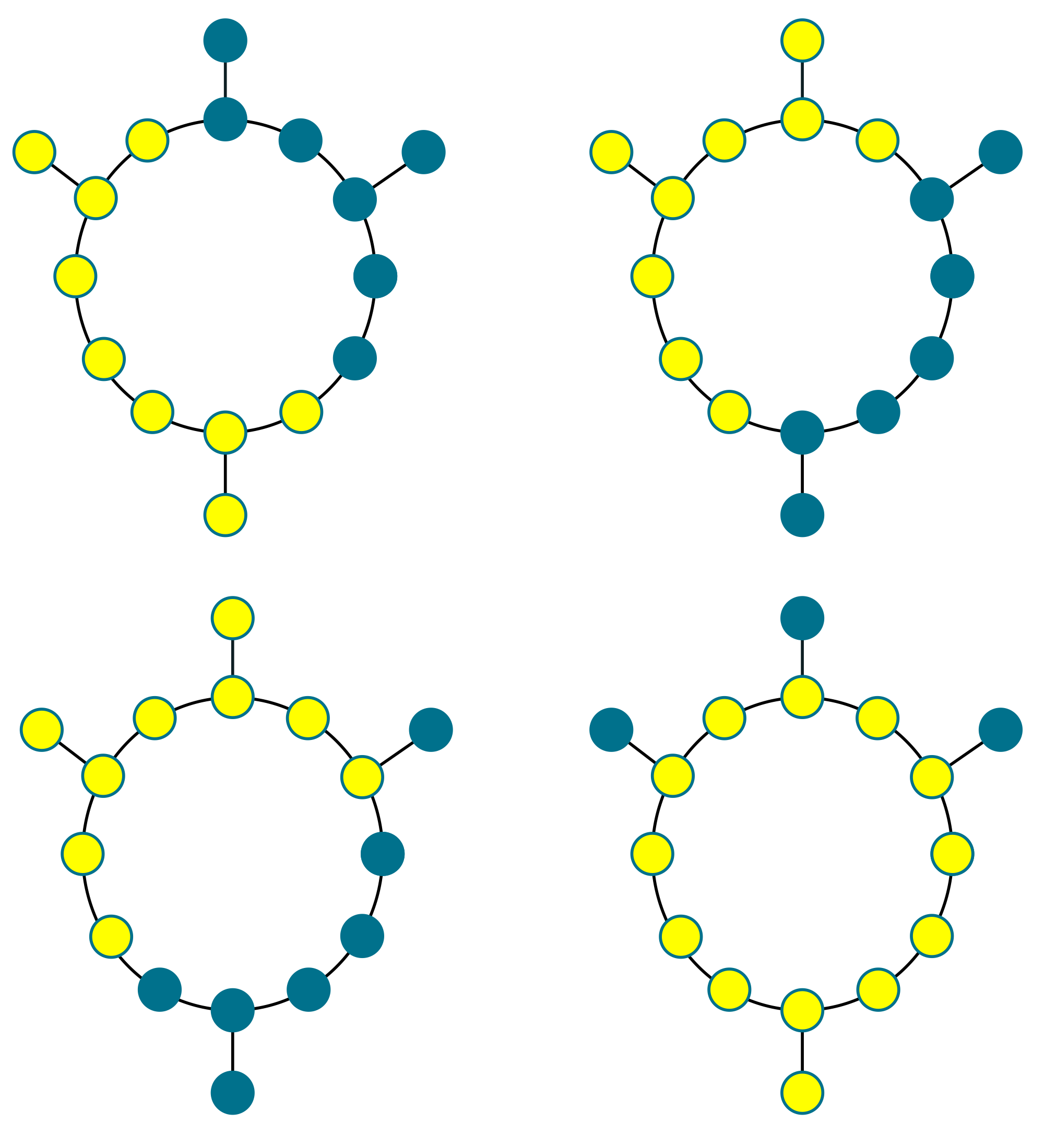}
            \caption{$4$-element covering}\label{fig:covering_2}
          \end{subfigure}
          
          \begin{subfigure}[b]{\linewidth}
          \centering
            \includegraphics[width=0.23\linewidth]{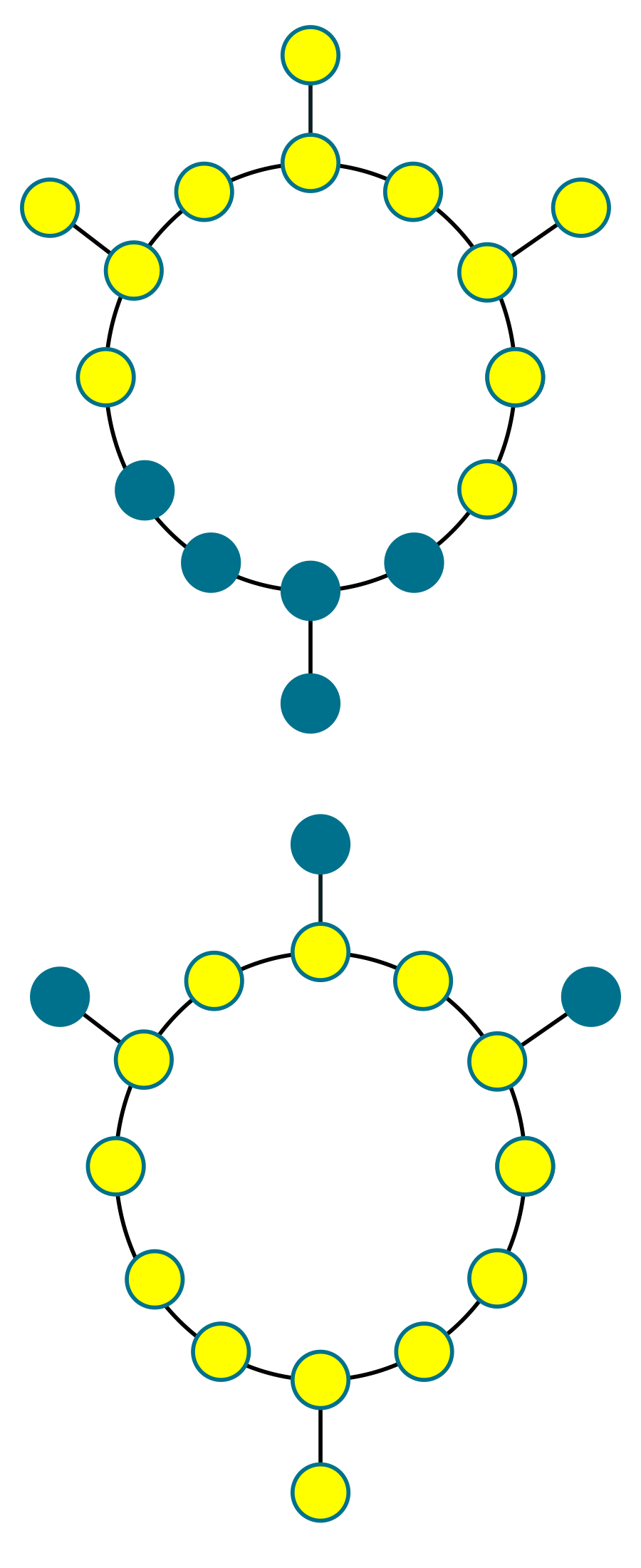}
            \caption{$2$-element covering}\label{fig:covering_1}
          \end{subfigure}
          \caption[Subarchitecture coverings]{Different subarchitecture coverings (marked as \tikz{\draw[fill=myyellow,line width=1pt]  circle(1ex);}) for $9$-qubit circuits on the \emph{ibmq\_guadalupe} architecture.}
          \label{fig:ibm-guadalupe-covering}\vspace*{3mm}
        \end{figure}

    \end{minipage}
      
        Finding a minimal subarchitecture covering (with respect to the number of elements) is actually very easy, since the
        entire architecture, per definition, covers all subarchitecture candidates---at the disadvantage of having the largest
        number of qubits. 
        In contrast, the set of subarchitecture candidates trivially forms the maximal covering---having the least
        number of qubits but the highest number of candidates. 
        In order to determine a subarchitecture covering that is as small as possible while still containing
        subarchitectures over as few qubits as possible, we propose the greedy algorithm as sketched in \autoref{alg:covering}.

        The algorithm starts off with the subarchitecture candidates as a covering and then iteratively picks a covering candidate (from smallest to largest).
        If the candidate covers more than a single architecture in the currently considered covering, the candidate replaces all covered architectures.
        This process is repeated until the size of the covering has been reduced to the desired size.
        This algorithm is guaranteed to terminate since, in the worst case, the entire architecture is returned.

        The bottleneck of this algorithm is the computation of the non-isomorphic subarchitectures, as it requires solving many instances of the (sub-)graph isomorphism problem. During the computation of the non-isomorphic subarchitectures, the set $\cand_n(A)$ can be constructed on the fly. Because the ordering $\prec$ is dependent on the sub-graph isomorphism relation, it can be deduced during the sub-graph isomorphism check of two subarchitectures. If it is determined that $A' \sqsubseteq A''$, then it is simply a matter of determining whether there are two nodes in $A''$ that are closer in $A''$ than in $A'$ to determine whether $A'' \prec A'$. The set $\mathcal{D}(A', A)$ can then be computed as the transitive closure.

        \begin{example}
          Once more, consider the \emph{ibmq\_guadalupe} architecture from the previous examples.
          Then, \autoref{fig:ibm-guadalupe-covering} shows a four- and a two-element subarchitecture covering for $9$-qubit circuits as determined by \autoref{alg:covering}.        		
          While the four-element covering shown in \autoref{fig:covering_2} contains smaller architectures (specifically, three
          $9$-qubit, and one $13$-qubit architecture), \autoref{fig:covering_1} shows that all three $9$-qubit architectures can be
          covered by a single $11$-qubit subarchitecture. 
        \end{example}
        
        Overall, \autoref{alg:covering} allows one to adaptively cover all the subarchitecture candidates---offering a trade-off
        between the number of subarchitectures to consider and their respective size. 
        The computed covering ensures that all practically relevant parts of the search space (i.e., excluding the
        pathological cases covered by \autoref{thm:iso-opt} and \autoref{thm:opt-informal}) are considered during the subsequent mapping. 
        
        \section{Resulting Tool}\label{sec:experiments}

      All of the above findings have been used to develop a Python-based tool that computes subarchitectures with a high coverage
      for a given architecture and size. In particular, this tool computes optimal subarchitecture candidates
      as discussed in \autoref{sec:opt-graphs} and minimal subarchitecture coverings using \autoref{alg:covering}.
      The tool is integrated into the quantum circuit mapping tool QMAP~(\mbox{\url{https://github.com/cda-tum/qmap}}), which is part of the \emph{Munich Quantum Toolkit} (MQT)~\cite{willeMQTQMAPEfficient2023}.
      All graph computations are performed using the retworkx
      library~\cite{treinishRetworkxHighperformanceGraph2022} which uses the
      VF2 algorithm~\cite{cordellaSubGraphIsomorphism2004} to solve the (sub-)graph isomorphism
      problem. Since the subgraph isomorphism problem is NP-hard~\cite{cookComplexityTheoremprovingProcedures1971},
      \autoref{alg:covering} is inherently an exponential algorithm (at least as long as the question of P and NP is not
      settled). To compute all non-isomorphic subarchitectures of an
      architecture, one needs to first compute all subarchitectures of which there are an exponential number with
      respect to the size of the whole architecture.
      To demonstrate the runtime of computing optimal subarchitectures and the effect of mapping to subarchitectures, experimental evaluations were conducted on a \SI{3.6}{\giga\hertz} Intel Xeon W-1370P machine running Ubuntu 20.04 with \SI{128}{\gibi\byte} of main memory. Each evaluation's time limit was set to 24 hours. 
      The current version of the tool performs all computations using a single thread. 
	  In principle, all non-isomorphic subarchitectures can be computed in parallel, which could accelerate the computation of optimal subarchitectures for larger architectures. 
	  In the sections that follow, we will show how the tool can be used to analyze state-of-the-art quantum computing architectures.
      \subsection{Optimal Subarchitectures and Coverings}

      The resulting tool was used to compute optimal subarchitecture candidates and minimal subarchitecture coverings for three
      representative quantum computing architectures: the $16$-qubit \emph{ibmq\_guadalupe} architecture, another $16$-qubit architecture consisting of two connected $8$-qubit rings that serve as the foundation of Rigetti's quantum computing architectures, and a $23$-qubit
      part of Google's Sycamore chip. More precisely, \autoref{tab:subarchitectures} lists the
      total number of connected subarchitectures and non-isomorphic subarchitectures for the three architectures,
      as well as the optimal subarchitecture candidates and minimal subarchitecture coverings for two different numbers
      of qubits for each device.

      First, results confirm the findings illustrated before concerning the ibmq\_guadalupe architecture. Since it is
      based on a $12$-qubit ring, the optimal candidates from $8$ qubits 
      onward are all at least of size $12$ because that is the point at which there are shorter connections between
      qubits on the ring. Therefore it is quite hard to reduce the number of qubits of the considered subarchitectures
      during mapping on the \emph{ibmq\_guadalupe}. However, because the search space in quantum circuit mapping is
      exponential in the number of allocated qubits, every qubit saved by using \autoref{alg:covering} significantly
      reduces the complexity of the mapping problem.

      Things get more interesting with the $16$-qubit Rigetti architecture which has a nice symmetry to it. The two
      central elements of this architecture are the $4$-qubit ring in the middle and the $8$-qubit rings attached to
      it. This is reflected in the covering graphs shown in \autoref{tab:subarchitectures}, which can be roughly
      separated into subgraphs containing the $4$-qubit ring and subgraphs containing the $8$-qubit ring. 

                      \begin{table*}[t]
                  \caption{Example subarchitectures}
          \label{tab:subarchitectures}
          \centering
          \resizebox{0.98\linewidth}{!}{ \tiny
          \begin{tabular}{m{.08\linewidth} m{0.2cm} >{\centering}m{1.6cm} >{\centering}m{1.6cm} m{.08\linewidth}
            m{.08\linewidth}  m{.08\linewidth} m{.08\linewidth}}
            \toprule
            Architecture & $|A|$ & Connected Subarchitectures & Non-isomorphic Subarchitectures  &  \multicolumn{2}{c}{Optimal Candidates} &
                                                                                                              \multicolumn{2}{c}{Coverings}
            \\ \midrule
            \multicolumn{4}{l}{ibmq\_guadalupe} & 8 Qubits & 10 Qubits & 8 Qubits & 10 Qubits\\
               \begin{minipage}{\linewidth}
                 \includegraphics[width=0.9\linewidth]{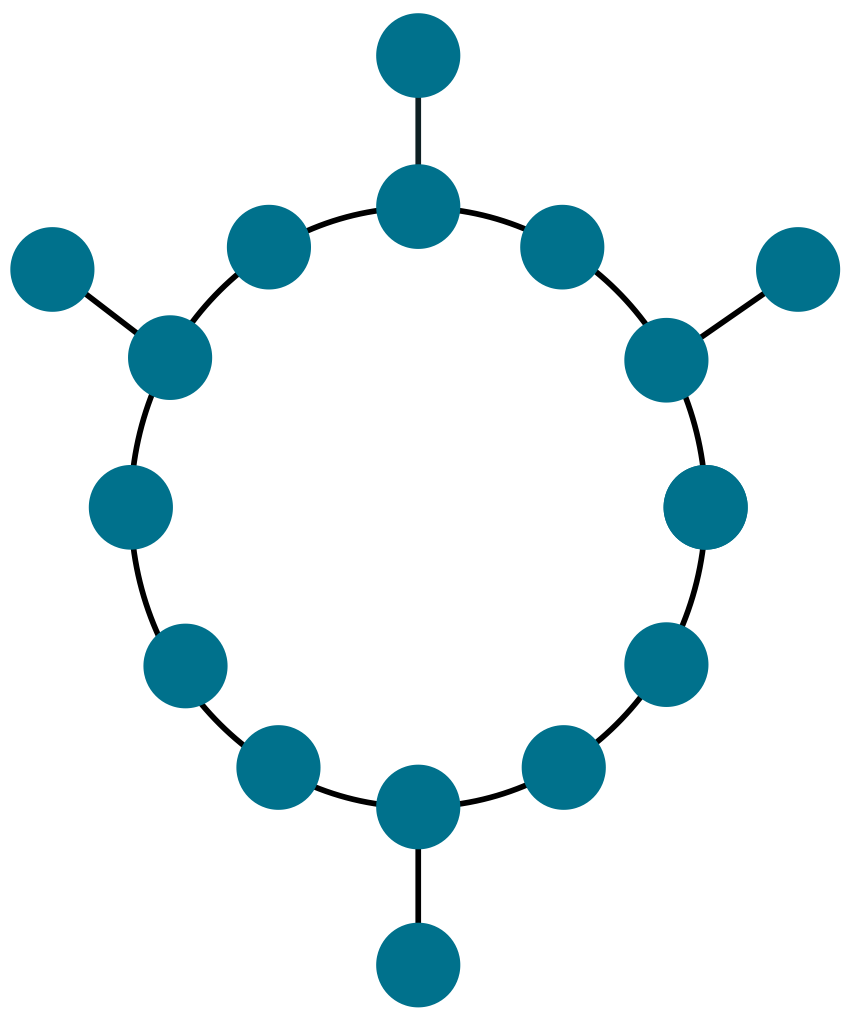}
                                   \end{minipage}& 16 & $746$ & $110$ & \begin{minipage}{\linewidth}
                                   \includegraphics[width=0.9\linewidth]{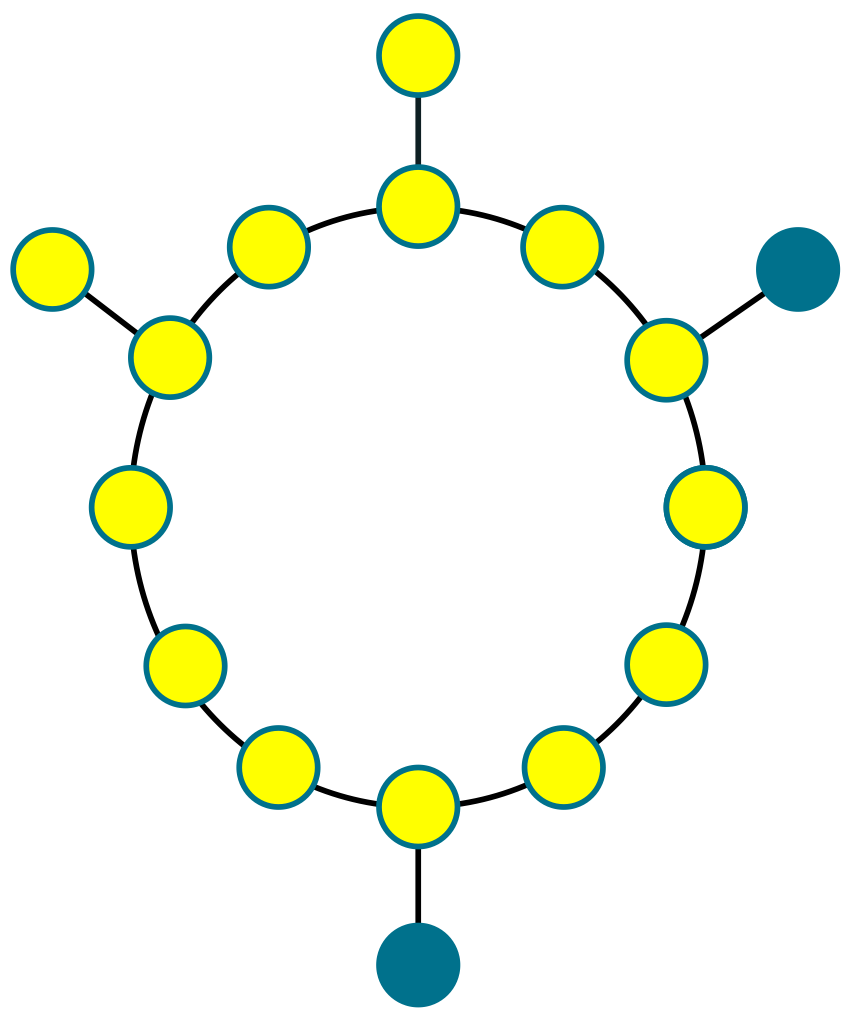}
                                 \end{minipage}            &
               \begin{minipage}{\linewidth}
                                   \includegraphics[width=0.9\linewidth]{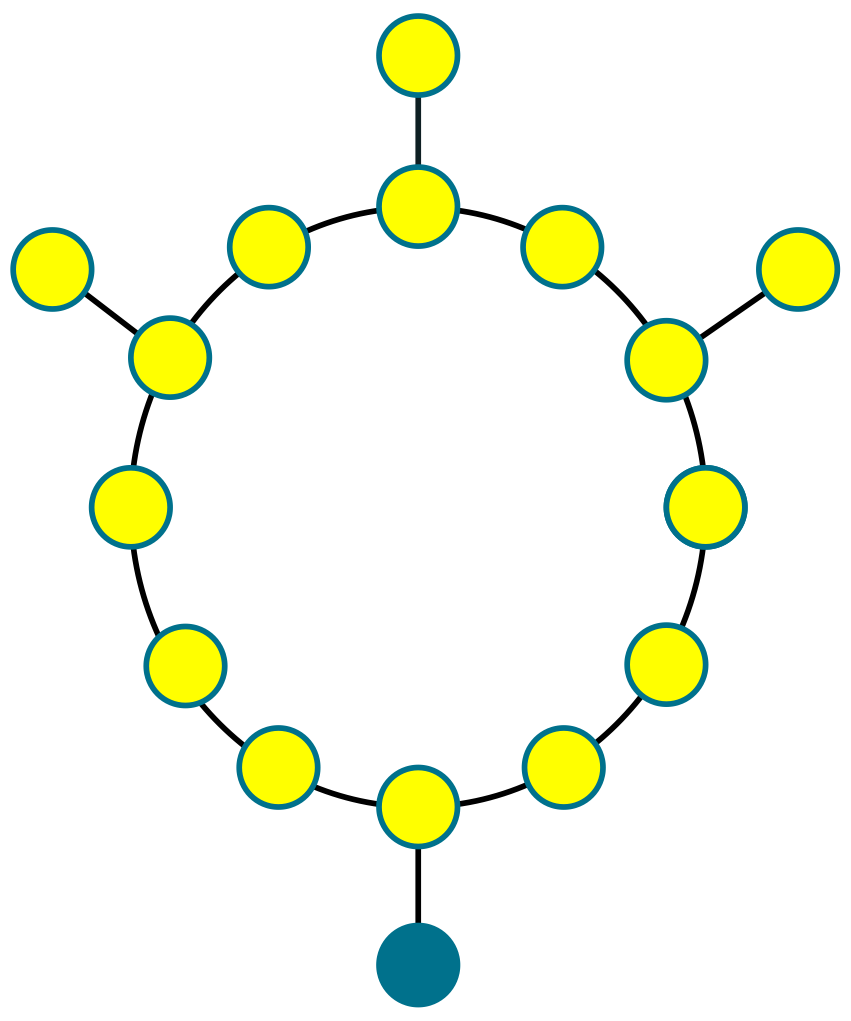}
                                 \end{minipage}&
                                                 \begin{minipage}{\linewidth}
                                                   \vspace*{0.25cm}
                                   \includegraphics[width=0.99\linewidth]{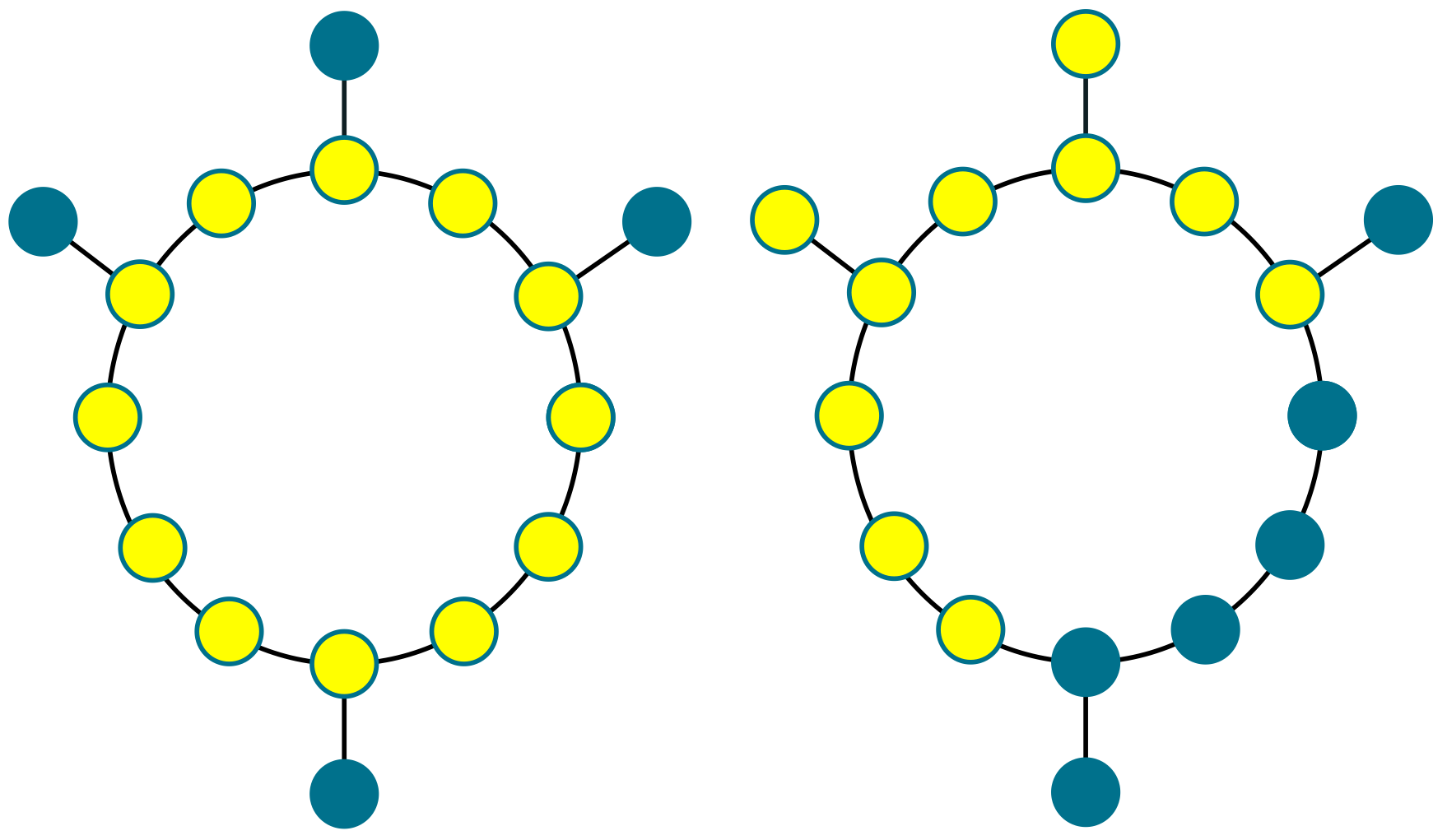}
                                   \end{minipage}
&
                                 \begin{minipage}{\linewidth}
                                   \includegraphics[width=0.99\linewidth]{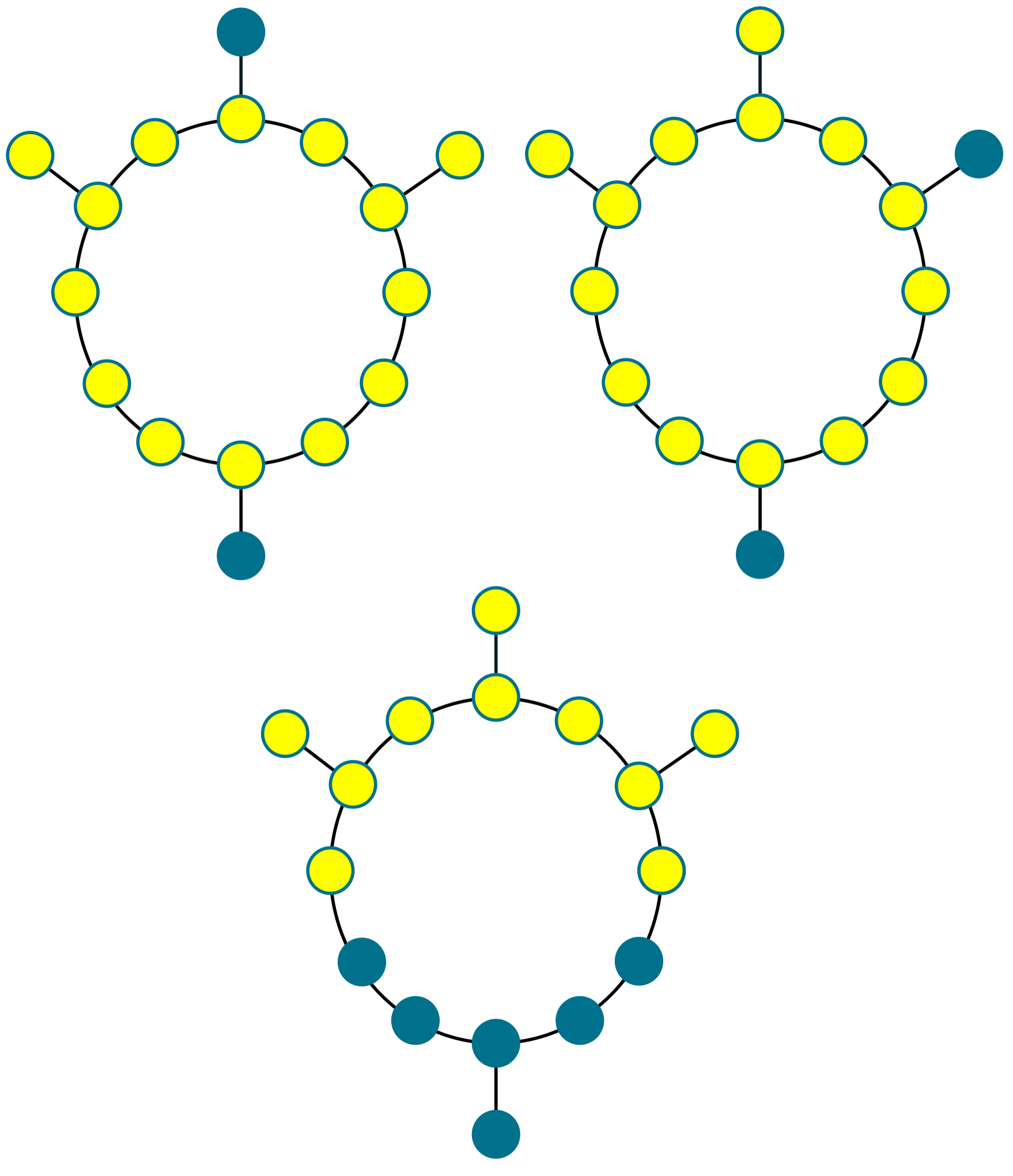}
                                 \end{minipage}\\\midrule
            \multicolumn{4}{l}{Rigetti} & 10 Qubits & 7 Qubits & 10 Qubits & 7 Qubits\\
               \begin{minipage}{\linewidth}
                 \includegraphics[width=0.9\linewidth]{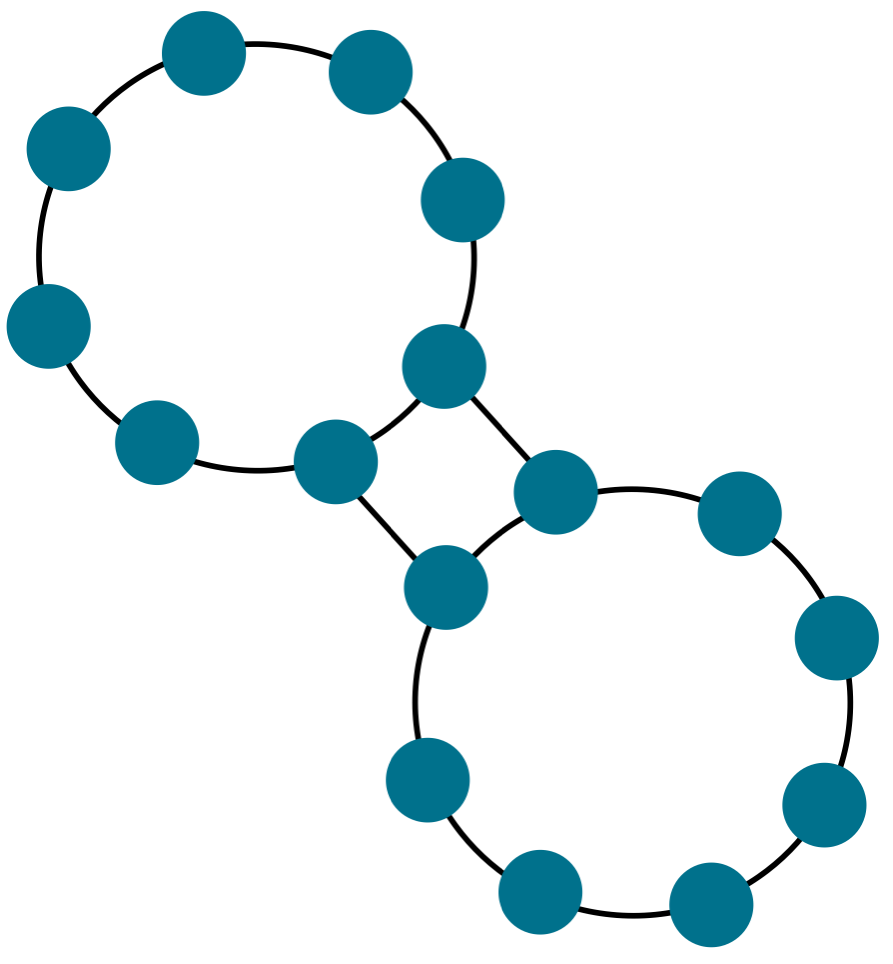}
                                   \end{minipage}& 16 & $1312$ & $184$ & \begin{minipage}{\linewidth}
                                   \includegraphics[width=0.9\linewidth]{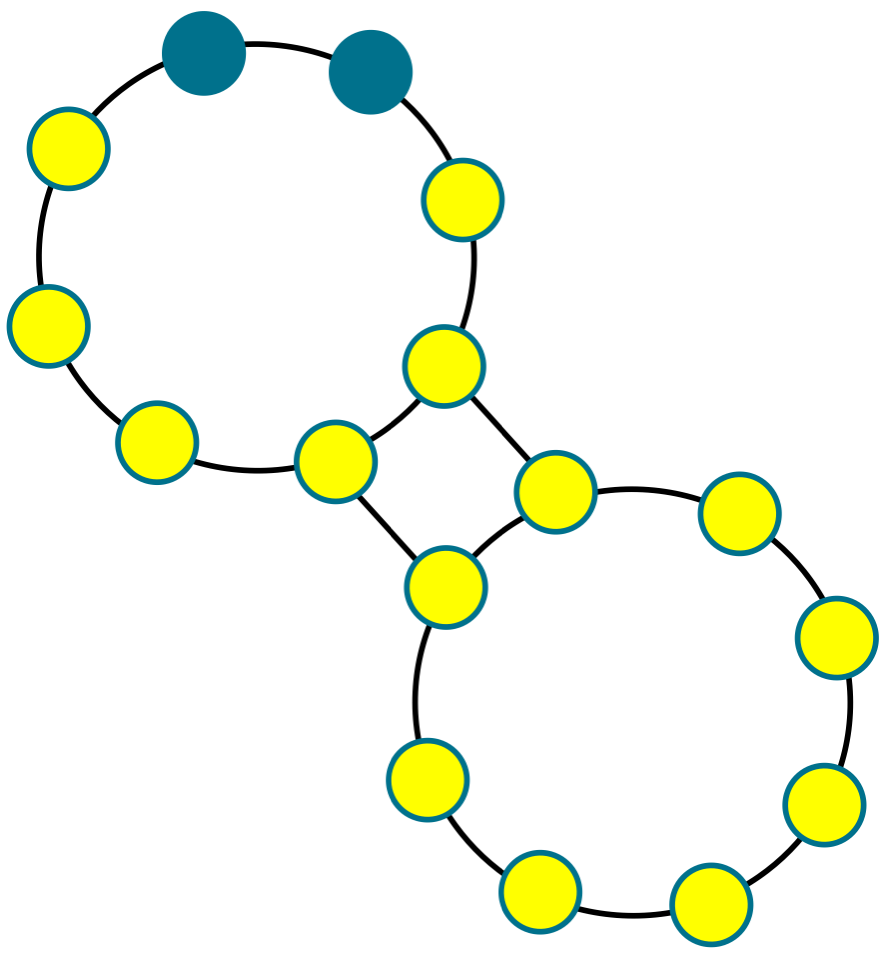}
                                 \end{minipage}            &
               \begin{minipage}{\linewidth}
                                   \includegraphics[width=0.9\linewidth]{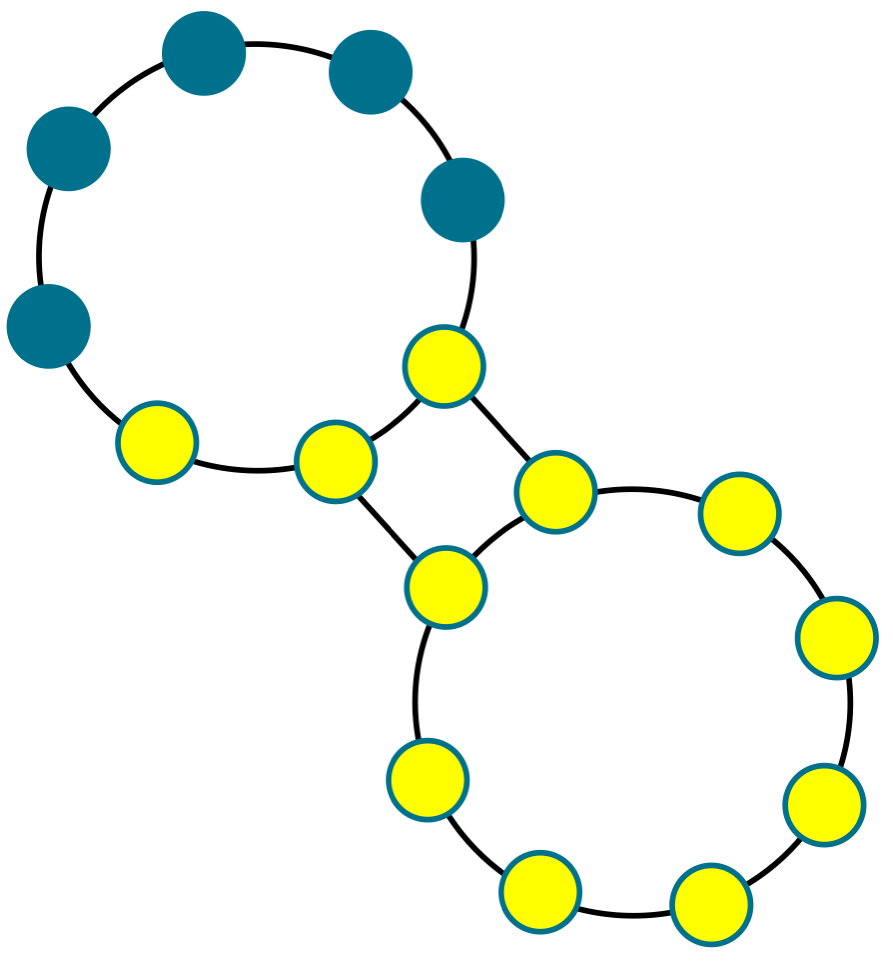}
                                 \end{minipage}&
                                 \begin{minipage}{\linewidth}
                                   \includegraphics[width=0.99\linewidth]{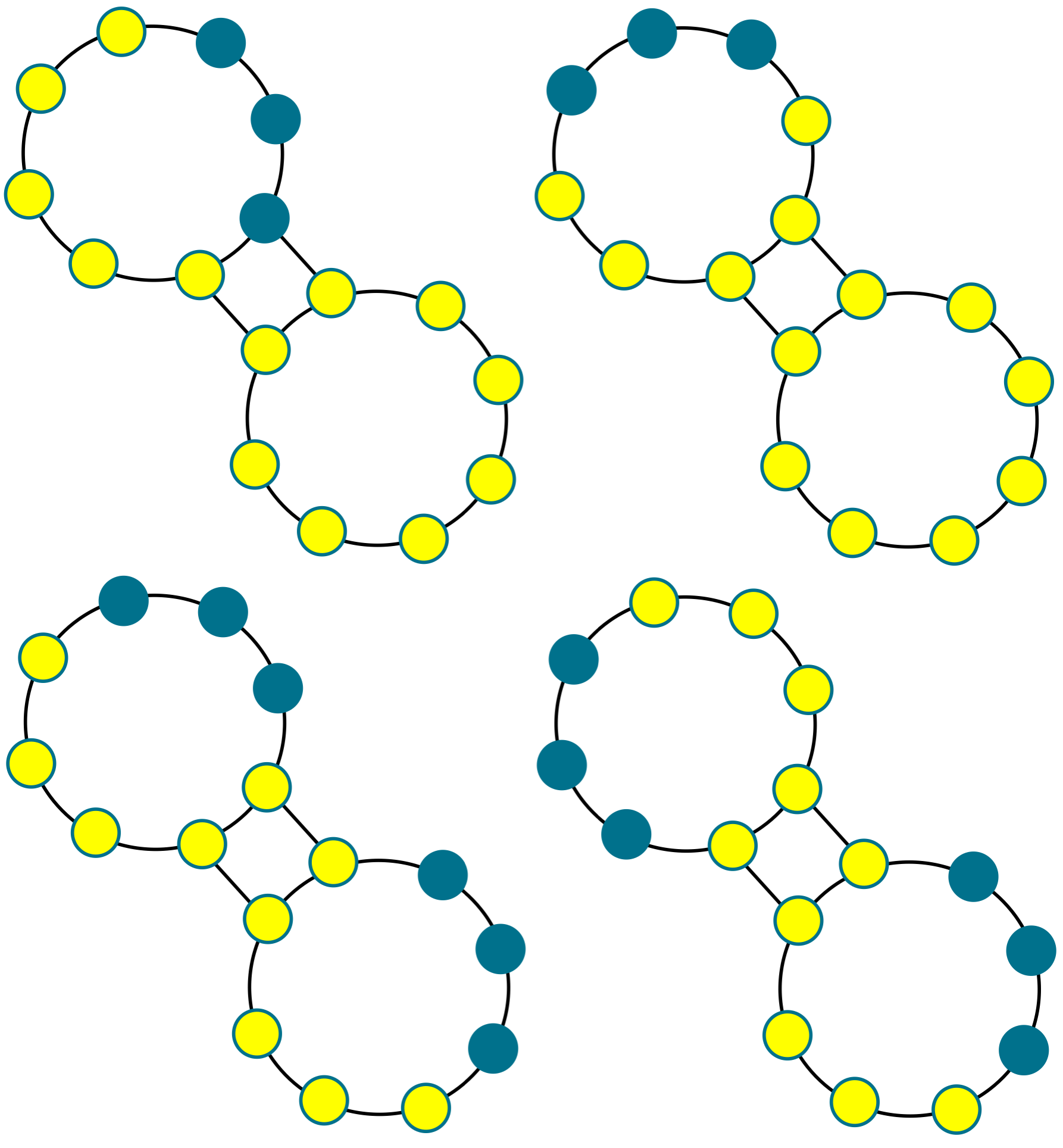}
                                   \end{minipage}
&
                                                    \begin{minipage}{\linewidth}
                                                      \vspace*{0.2cm}
                                   \includegraphics[width=0.99\linewidth]{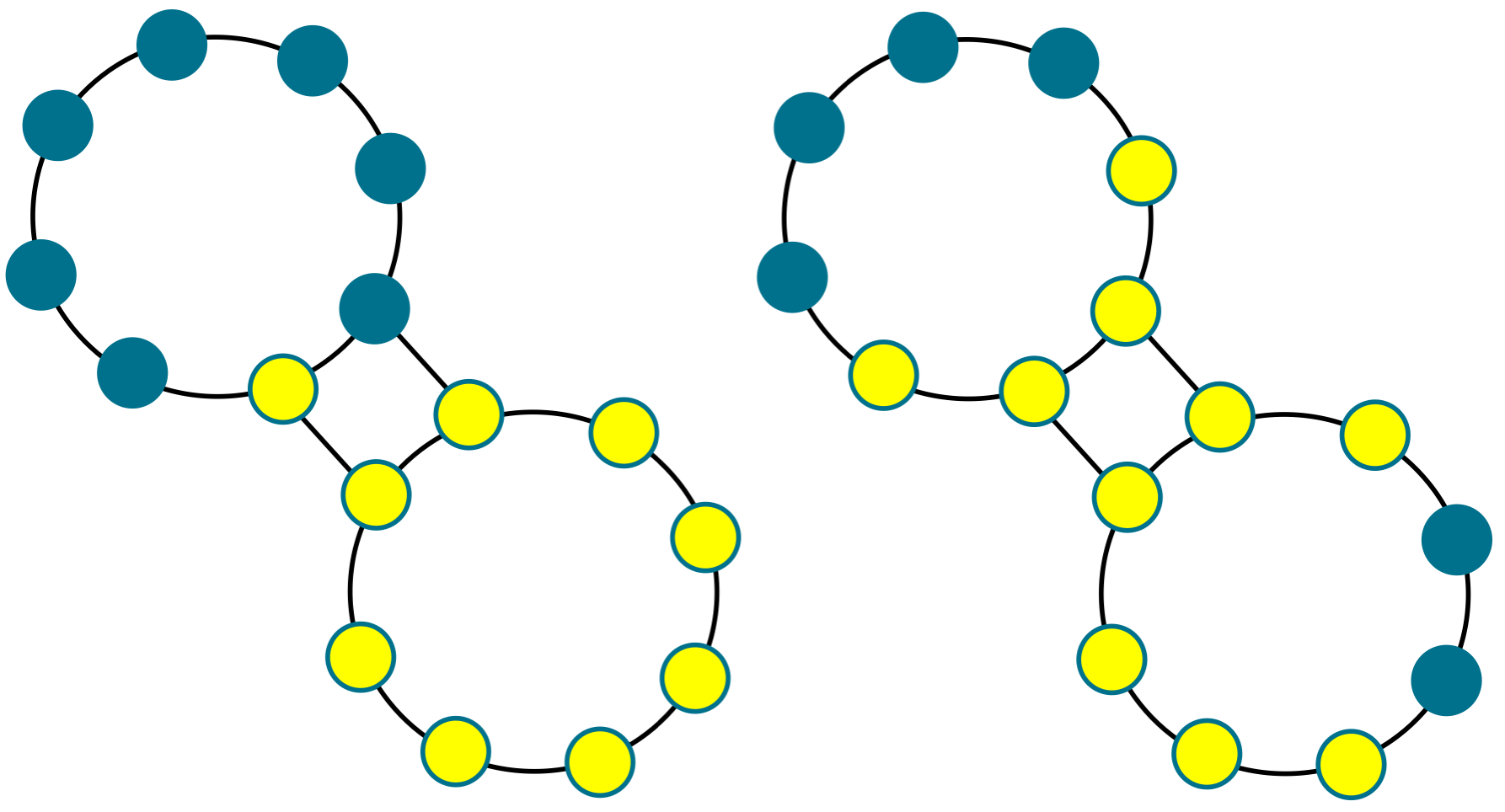}
                                 \end{minipage}\\ \midrule
            \multicolumn{4}{l}{Sycamore} & 7 Qubits & 13 Qubits & 7 Qubits & 13 Qubits\\
            \begin{minipage}{\linewidth}
              \includegraphics[width=0.9\linewidth]{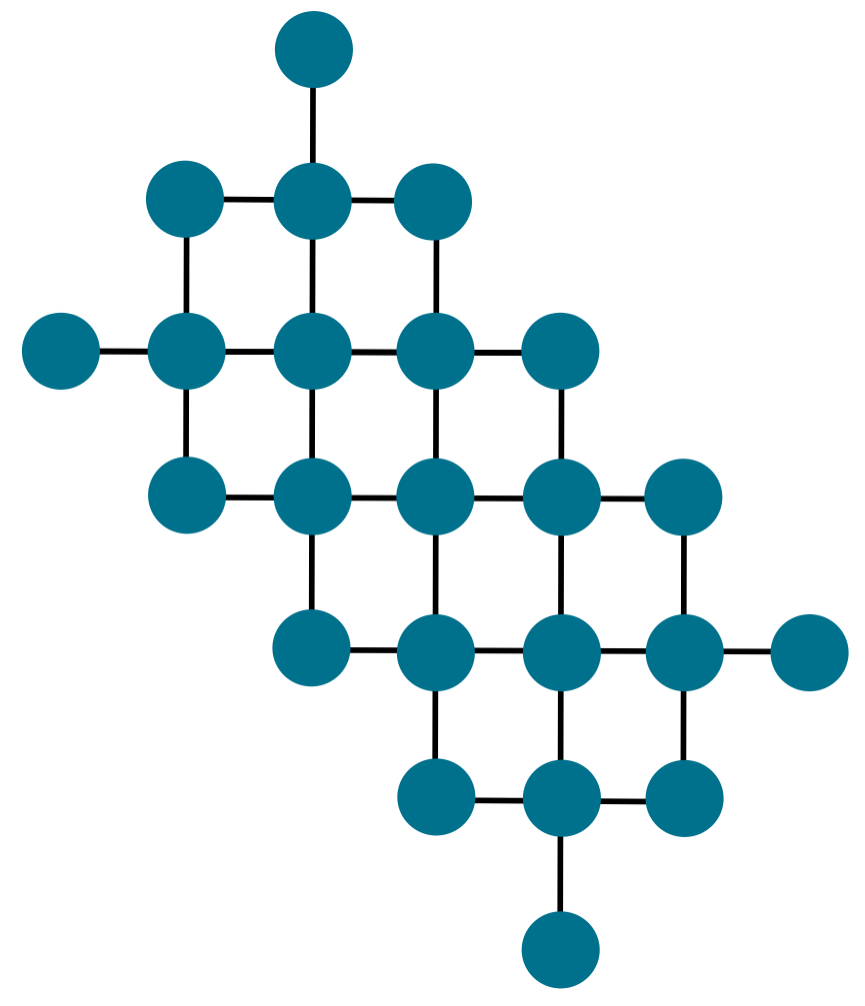}
            \end{minipage}& 23 & $300015$ & $24786$ & \begin{minipage}{\linewidth}
                                   \includegraphics[width=0.9\linewidth]{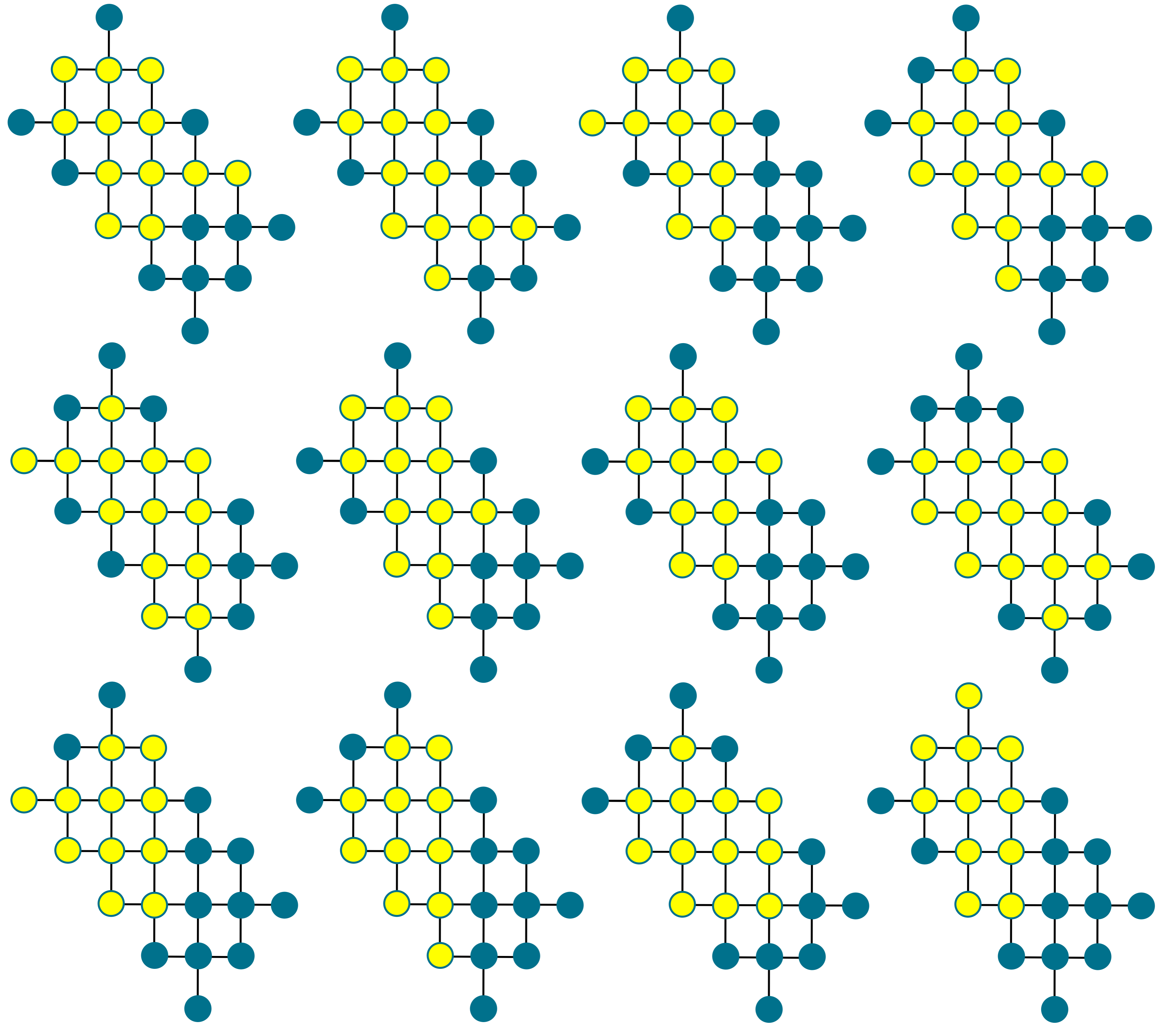}
                                 \end{minipage}            &
               \begin{minipage}{\linewidth}
                                   \includegraphics[width=0.9\linewidth]{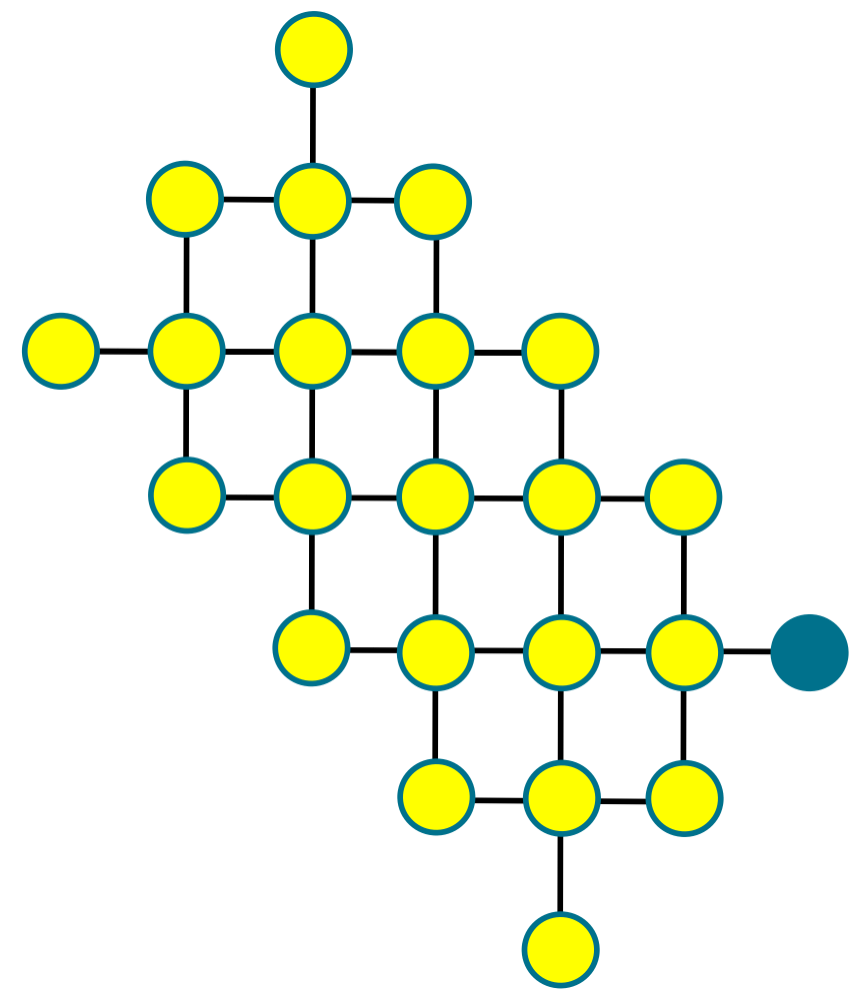}
                                 \end{minipage}&
                                 \begin{minipage}{\linewidth}
                                   \includegraphics[width=0.9\linewidth]{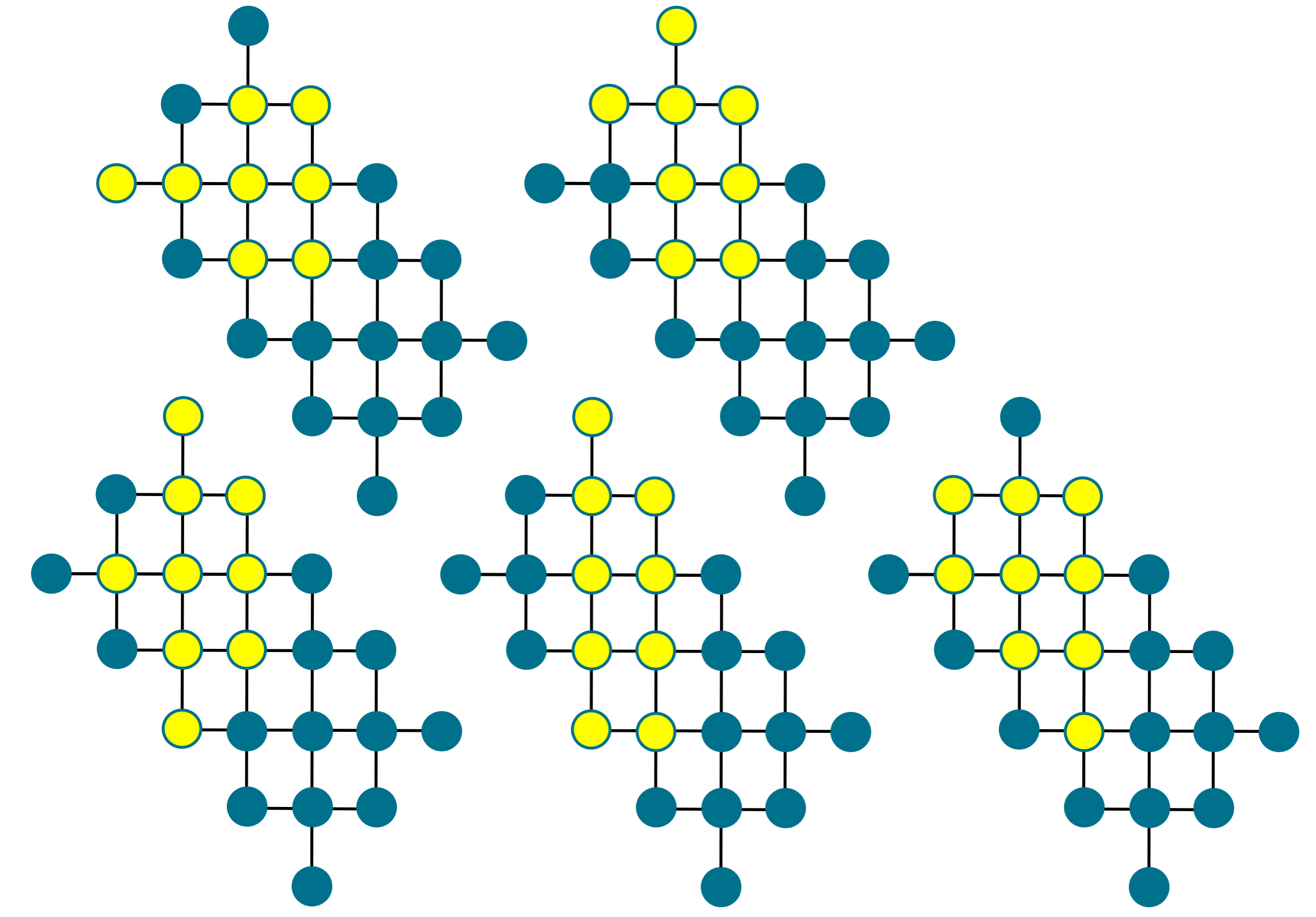}
                                   \end{minipage}
&
                                 \begin{minipage}{\linewidth}
                                   \begin{tikzpicture}[overlay]
                                     \node[] at (0.8,-0.1) {$\mathbf{\vdots} ~ +95$};
                                   \end{tikzpicture}
                                   \includegraphics[width=0.9\linewidth]{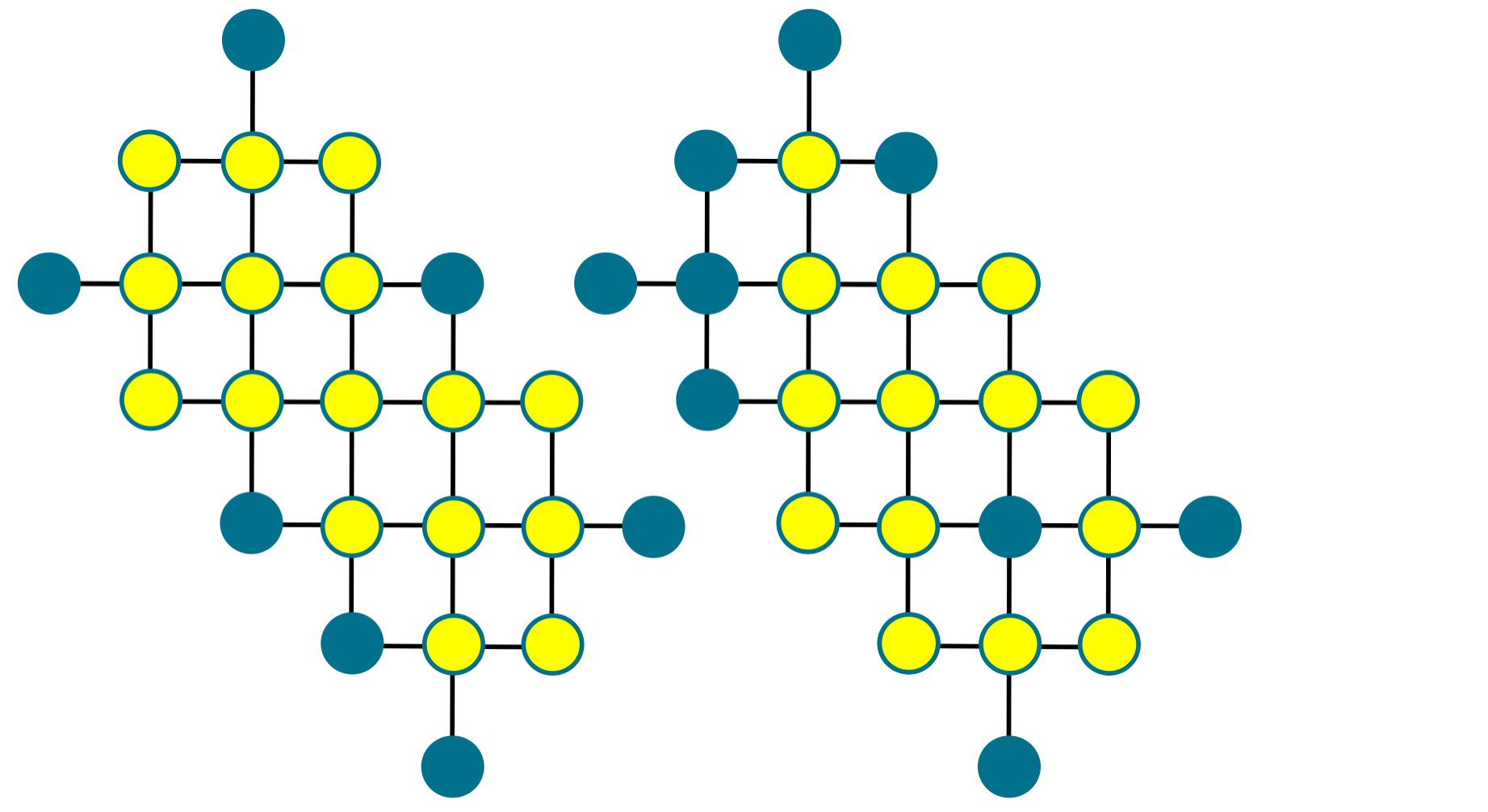}
                                   \end{minipage}\\
\bottomrule
          \end{tabular}}
      \end{table*}

      Finally, the Sycamore architecture by Google is an extreme case as it has a higher connectivity than any other
      superconducting quantum circuit architecture considered here. This is reflected in the large number of \mbox{non-isomorphic}
      subarchitectures contained in the $23$-qubit architecture. Because of this high connectivity, one would expect
      that, according to \autoref{thm:opt-layout}, subarchitectures can be improved by adding additional qubits to
      complete the many $4$-qubit rings present in the Sycamore architecture. And, indeed, it can be observed that, even for $13$ qubits, the optimal subarchitecture candidate uses all but one of the architecture's qubits. Furthermore, it is hard to compute a small set of covering architectures.
      This makes the trade-off between coverage of optimal mapping solutions, complexity, and number of subarchitectures that need to be considered
      during quantum circuit mapping tricky. 
      Even for the $13$ qubit case, no covering with less than $97$ elements can be found that does not contain a subarchitecture with more
      than $18$ qubits. However, this is still a significant reduction from the $1153$ desirable subarchitectures for $13$
      qubits.

      Overall, \autoref{tab:subarchitectures} demonstrates just how hard it can be to reduce the size of
      subarchitectures considered for qubit allocation. This fact makes the theoretical findings provided in this work and the resulting tool
      especially appealing because they aid in understanding and tackling this difficult problem and provide a solid
      foundation for handling the enormous search space for future quantum circuit mapping techniques.

      \subsection{Computing the Partial Order $\prec$}
      \label{sec:runtime_arch}
      \begin{table}
          \sisetup{round-mode=places, round-precision=3}
        \caption{Construction of the partial order $\prec$ for architectures of different sizes}\label{tab:construction}
        \begin{tabular}[t]{l r r r r r}
    \toprule
    Architecture & $|A|$ & Connected & Non-isomorphic & $t$~[\si{s}]\\ \midrule
          \begin{filecontents*}{lib_gen.csv}
Architecture,nqubits,nonisomorphic,isomorphic,t
ibmq-lima,5,17,6,0.001978
ibm-lagos,7,37,9,0.003426
ibmq-guadalupe,16,746,110,1.379794
ibmq-toronto,27,43721,5197,2544.953877
sycamore-1-ring,4,13,4,0.000191
sycamore-2-ring,6,40,9,0.001185
sycamore-3-ring,8,117,24,0.011962
sycamore-4-ring,9,218,28,0.020113
sycamore-5-ring,11,609,111,0.367660
sycamore-6-ring,13,1852,276,2.406237
sycamore-7-ring,14,3343,383,4.592274
sycamore-8-ring,16,9118,2004,93.152393
sycamore-9-ring,18,27435,4156,402.590249
sycamore-10-ring,19,49420,5364,665.921583
sycamore-full,23,300015,24786,14240.843003
rigetti-1-ring,8,57,8,0.193432
rigetti-2-ring,16,1312,184,2.176312
rigetti-3-ring,24,31033,6181,1703.184524

\end{filecontents*}
\csvreader[head to column names]{lib_gen.csv}{}{\\\Architecture&\nqubits&\nonisomorphic&\isomorphic&\num{\t}}\\
                                                                                                               rigetti-4-ring & 32 & - &- & \emph{timeout}
    \\\bottomrule
  \end{tabular}
\end{table}

      To showcase the runtime scaling of computing optimal subarchitectures, consider \autoref{tab:construction}.
      It illustrates the runtime with respect to the size of a given architecture and its number of different subarchitectures (isomorphic as well as non-isomorphic) when constructing the initial ordering $\prec$ of the subarchitectures---the computationally most expensive part of \autoref{alg:covering}.
      The considered architectures are comprised of different IBM quantum computers as well as versions of the Rigetti and Google architectures considered in the previous section.
      The Rigetti and Google architectures are comprised of 8- and 4-qubit rings, respectively, which are connected on a regular grid.
      For these architectures the number of connected rings for which $\prec$ has been computed, were varied to show the runtime dependency.

      Architectures with more than 30 qubits could not be evaluated within 24 hours using the current version of the tool. This is mainly due to the current implementation not taking advantage of all available hardware resources, e.g., via parallelization. While there is naturally still a limit to the size of the considered architectures, a more optimized implementation can help analyze larger and more intricate architectures.
      The runtimes also suggest an optimization when mapping smaller circuits to regularly repeating architectures: instead of computing $\prec$ on the entire architecture, one can pick a smaller subarchitecture that is likely to contain most of the relevant subarchitectures for the circuit in question.
      This way, one can hope to gain a significant speedup whithout losing out on too many subarchitectures, as the runtime increases more drastically with the size of the architecture than the number of non-isomorphic subarchitectures.

      \subsection{Impact on Quantum Circuit Mapping}
      \label{sec:runtime_mapping}
      \begin{table*}
          \sisetup{round-mode=places, round-precision=3}
  \caption{Optimal mapping with different subarchitectures}\label{tab:mapping}
\vspace*{-2mm}
    \begin{subtable}[t]{0.48\linewidth}\centering
    \caption{4-qubit circuits on 6-qubit architecture}
    \resizebox{\linewidth}{!}{\small \begin{tabular}[t]{p{0.4\linewidth}>{\raggedleft}p{0.3\linewidth}>{\raggedleft\arraybackslash}p{0.3\linewidth}}
    \toprule
          Subarchitecture & \#SWAP & $t$~[\si{s}]\\\midrule
                                       \multicolumn{3}{c}{Two-local-random}
          \begin{filecontents*}{twolocalrandom.csv}
arch,t,swaps
ring,0.289591,9
line,1.076683,12
fork,0.054436,8
cover,3.228464,8
full,11.798218,8
\end{filecontents*}
\csvreader[head to column names]{twolocalrandom.csv}{}{\\\arch&\swaps&\num{\t}}\\\midrule
                                       \multicolumn{3}{c}{Portfolio VQE}
          \begin{filecontents*}{portfoliovqe.csv}
arch,t,swaps
ring,0.251028,9
line,0.554324,12
fork,0.03791,8
cover,0.294058,8
full,12.620578,8
\end{filecontents*}
\csvreader[head to column names]{portfoliovqe.csv}{}{\\\arch&\swaps&\num{\t}}\\\midrule
                                       \multicolumn{3}{c}{QFT Entangled}
          \begin{filecontents*}{qftentangled.csv}
arch,t,swaps
ring,0.105856,3
line,0.757779,5
fork,0.04277,6
cover,0.291615,5
full,1.601686,3
\end{filecontents*}
\csvreader[head to column names]{qftentangled.csv}{}{\\\arch&\swaps&\num{\t}}\\\midrule
                                       \multicolumn{3}{c}{Realamprandom}
          \begin{filecontents*}{realamprandom.csv}
arch,t,swaps
ring,0.449355,9
line,2.56265,12
fork,0.050433,8
cover,1.556439,8
full,4.3586,8
 
\end{filecontents*}
\csvreader[head to column names]{realamprandom.csv}{}{\\\arch&\swaps&\num{\t}}\\
\bottomrule
  \end{tabular}}
  \end{subtable}\hspace*{0.01\linewidth}
    \begin{subtable}[t]{0.48\linewidth}\centering
    \caption{5-qubit circuits on 7-qubit architecture}
    \resizebox{\linewidth}{!}{\small \begin{tabular}[t]{p{0.4\linewidth}>{\raggedleft}p{0.3\linewidth}>{\raggedleft\arraybackslash}p{0.3\linewidth}}
    \toprule
          Subarchitecture & \#SWAP & $t$~[\si{s}]\\ \midrule
          \multicolumn{3}{c}{Graphstate}
          \begin{filecontents*}{graphstate_5.csv}
arch,t,swaps
ring,0.055281,1
line,0.085661,3
fork,0.033517,2
cover,0.138259,2
full,61.212605,1
\end{filecontents*}
\csvreader[head to column names]{graphstate_5.csv}{}{\\\arch&\swaps&\num{\t}}\\\midrule
          \multicolumn{3}{c}{AE}
          \begin{filecontents*}{ae_5.csv}
arch,t,swaps
ring,2.424407,6
line,1.495007,6
fork,0.581055,4
cover,5.259435,4
full,227.778268,4
\end{filecontents*}
\csvreader[head to column names]{ae_5.csv}{}{\\\arch&\swaps&\num{\t}}\\\midrule
          \multicolumn{3}{c}{QPE Exact}
          \begin{filecontents*}{qpeexact_5.csv}
arch,t,swaps
ring,79.958769,6
line,10.944864,7
fork,0.598511,5
cover,115.252387,5
full,1189.600944,5
\end{filecontents*}
\csvreader[head to column names]{qpeexact_5.csv}{}{\\\arch&\swaps&\num{\t}}\\\midrule
          \multicolumn{3}{c}{QGAN}
          \begin{filecontents*}{qgan_5.csv}
arch,t,swaps
ring,0.873071,6
line,0.883597,6
fork,0.266946,4
cover,1.482059,4
full,171.094740,4
\end{filecontents*}
\csvreader[head to column names]{qgan_5.csv}{}{\\\arch&\swaps&\num{\t}}\\ \bottomrule
  \end{tabular}}
\end{subtable}
\vspace*{-5mm}
\end{table*}
      In order to illustrate the potential impact of considering optimal subarchitectures for quantum circuit mapping, \autoref{tab:mapping} shows the effects of subarchitectures on the runtime and number of SWAPs when mapping quantum circuits.
      To this end, the exact mapper available in QMAP~(\mbox{\url{https://github.com/cda-tum/qmap}})~\cite{willeMappingQuantumCircuits2019} was used for exact mapping, and the benchmarks were taken from the benchmark library MQTBench (version 0.2.2)~\cite{quetschlichMQTBenchBenchmarking2022}.

      As architectures, a the 6- and 7-qubit architectures in \autoref{fig:benchmark_rings} were considered. Both are comprised of a central ring of qubits with one additional qubit connected to one of the qubit rings. This way, the number of different non-isomorphic subarchitectures is kept small while seeing the effect of having shorter connections, different subarchitectures, and covering subarchitectures. There are only 5 subarchitectures to consider: the central ring (\emph{ring}), the 4- (5-) qubit line (\emph{line}), the 3- (4-) qubit line with an additional qubit connected to one of the middle qubits (\emph{fork}), the subarchitecture containing both \emph{line} and \emph{fork} (\emph{cover}), and the entire architecture itself (\emph{full}).
      The optimal subarchitecture for 4- and 5-qubit circuits on both architectures is the entire architecture itself.
      However, from \autoref{tab:mapping}, we can see that mapping to the entire architecture is often quite costly in terms of runtime compared to mapping to subarchitectures.
      A 2-subarchitecture covering can be obtained by taking the \emph{ring} and \emph{cover} subarchitectures together. This often yields much better runtimes without losing the optimal solution because it is guaranteed to be contained in these two subarchitectures by \autoref{thm:opt-informal}.

            \begin{figure}[t]
        \centering
        \begin{subfigure}[b]{0.24\linewidth}
          \includegraphics[width=\linewidth]{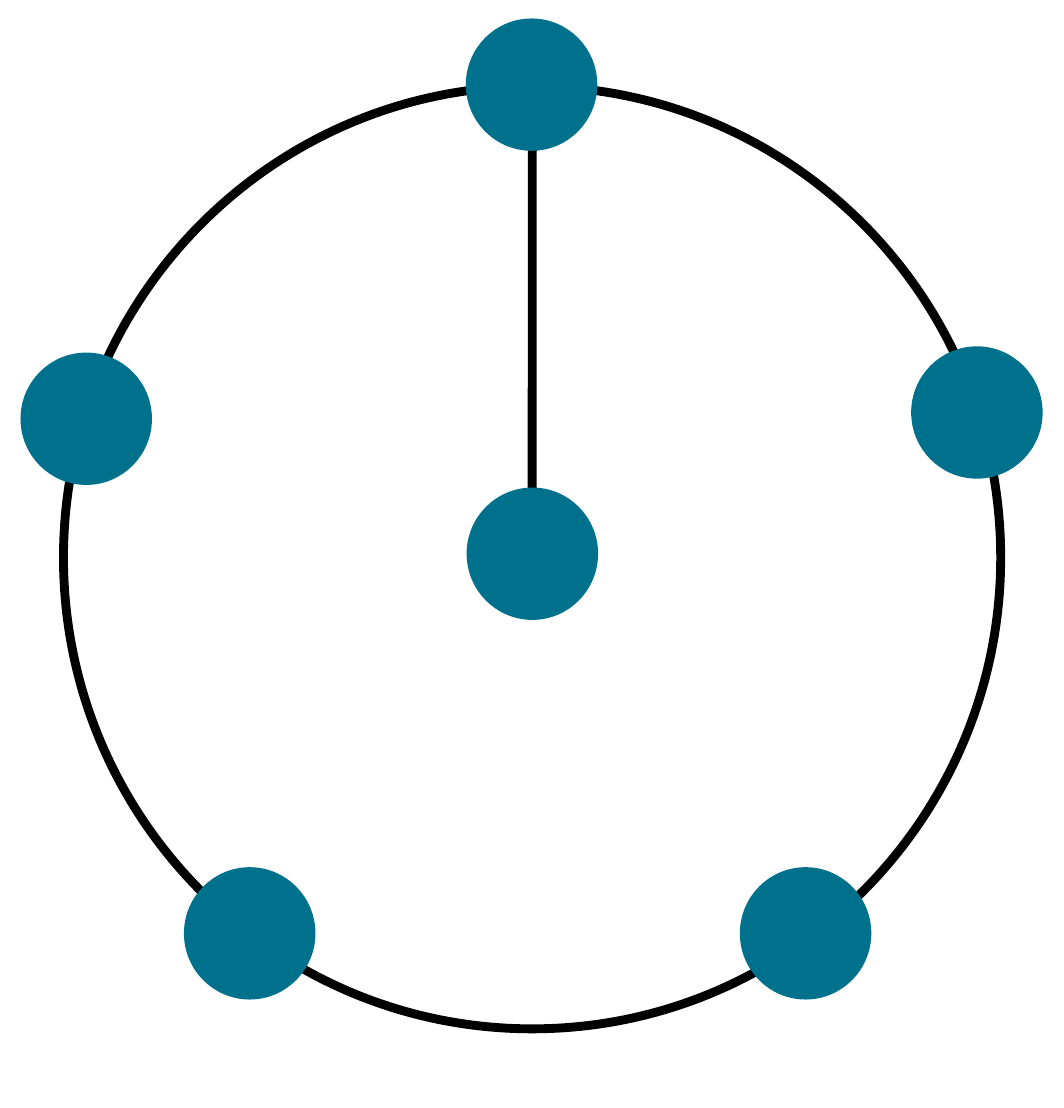}
          \end{subfigure}\hspace*{.2\linewidth}
        \begin{subfigure}[b]{0.24\linewidth}
          \includegraphics[width=\linewidth]{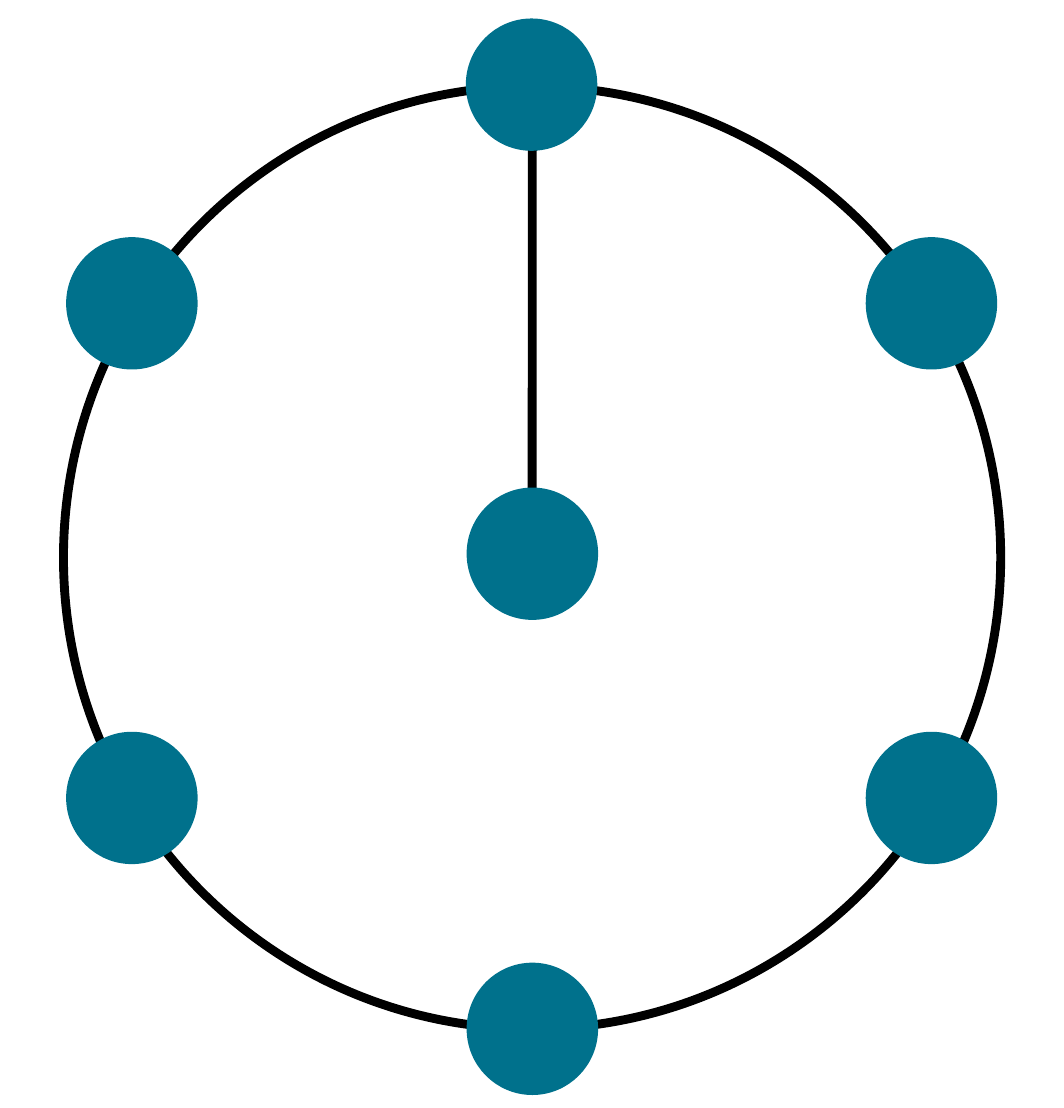}
        \end{subfigure}
        \caption{Architectures for optimal quantum circuit mapping}
        \label{fig:benchmark_rings}
      \end{figure}
      
      This table also shows the validity of \autoref{thm:opt-layout}, as the additional qubit completing the ring actually helps to improve the quality of the mapping. According to \autoref{thm:iso-opt}, there are circuits that are cheaper to map to \emph{cover} than to the two smaller architectures \emph{line} and \emph{fork}. This theoretical property was not observed in the experiments as the considered real-world quantum circuits where way too shallow and structured to possibly benefit from being mapped to \emph{cover}.

        \section{Conclusion and Future Work}\label{sec:conclusion}

        In this work, we introduced the notion of optimal subarchitectures for mapping $n$-qubit quantum circuits
        and disproved a previous conjecture that all $n$-qubit quantum circuits can be mapped to some $n$-qubit
        subarchitecture of a quantum computing device without potentially eliminating optimal mapping solutions. In
        fact, quite the opposite is the case: trying to reduce the number of qubits considered in the qubit allocation
        process without cutting off essential parts of the search space in the subsequent
        mapping is incredibly difficult. Despite this theoretical result, the structure of the quantum circuits that require such large optimal
        subarchitectures is pretty artificial, and the conditions for optimality can be relaxed a bit. Hence, we
        introduced an algorithm for computing subarchitectures that constitutes a trade-off between coverage of
        optimality and architecture complexity. The resulting tool is integrated into the open-source tool QMAP (available at \url{https://github.com/cda-tum/qmap}), which is part of the \emph{Munich Quantum Toolkit} (MQT).

        Based on this first method for computing near-optimal subarchitectures, there are several possible directions
        for improvement:

        \begin{itemize}
        \item Quantum circuits used in real-world applications naturally possess a lot of structure. It
        might, therefore, be possible to compute optimal subarchitectures for certain \emph{classes} of quantum circuits
        instead of any quantum circuit of a given size.  These classes can probably be
        deduced from two-qubit interaction patterns repeatedly found in real-world quantum circuits.
        \item         In this work, no distinction was made between isomorphic subarchitectures of a quantum computing device.
        In reality, neither the qubits nor the connections between them on an architecture are equally reliable.
        Consequently, placing a certain subarchitecture on a different part of the whole architecture might yield more reliable circuit executions.
        Tools like \emph{mapomatic}~\cite{nationMapomatic} exist that search for low-noise subarchitectures given an already-compiled quantum circuit.
        Such information could additionally be included in the methodology for determining suitable subarchitectures to provide for quantum compilers that consider
        noise.
      \item         In order to compute (near-)optimal subarchitectures of future large-scale quantum computers, more
        efficient algorithms than the one initially presented here will have to be developed. A promising approach is to
        take advantage of the highly symmetric structure of real-world quantum computing architectures. 
          
        \end{itemize}

        To summarize, this work has laid the groundwork for further research into the problem of computing (near-)optimal subarchitectures of state-of-the-art quantum computers. 
        This first work is a good starting point for more methods and improvements that can help develop quantum circuit mappers that can handle the mapping problem for large-scale quantum computers in an efficient way. 

\subsection*{Acknowledgements}
This work received funding from the European Research Council (ERC) under the European Union’s Horizon 2020 research and innovation program (grant agreement No. 101001318), was part of the Munich Quantum Valley, which is supported by the Bavarian state government with funds from the Hightech Agenda Bayern Plus, and has been supported by the BMWK on the basis of a decision by the German Bundestag through project QuaST, as well as by the BMK, BMDW, and the State of Upper Austria in the frame of the COMET program (managed by the FFG).

        \printbibliography

\end{document}